\documentclass[format=manuscript, review=false, nonacm = true]{acmart}
\usepackage{xcolor}
\usepackage{mathtools}
\usepackage{amsfonts}
\usepackage{amsthm}
\usepackage{multirow}
\usepackage[colorinlistoftodos,prependcaption, disable]{todonotes}

  \usepackage[normalem]{ulem}

  \newtheorem{problem}{Problem}
  \newtheorem{assumption}{Assumption}
  \newtheorem*{remark}{Remark}
                
\begin{document}
\title{On the cluster admission problem for cloud computing}

\titlenote{An early version of this work has appeared as a 6-page extended abstract in the proceedings of the 14th Workshop on the Economics of Networks, Systems and Computation (NetEcon'19).}
\author{Ludwig Dierks}
\affiliation{%
	\institution{University of Zurich}}
\email{dierks@ifi.uzh.ch}
\author{Ian A. Kash}
\affiliation{%
	\institution{University of Illinois at Chicago}
}
\email{	iankash@uic.edu}
\author{Sven Seuken}
\affiliation{%
	\institution{University of Zurich}}
\email{seuken@ifi.uzh.ch}
\authorsaddresses{Authors' email addresses: Ludwig Dierks, dierks@ifi.uzh.ch; Ian A. Kash, iankash@uic.edu; Sven Seuken, seuken@ifi.uzh.ch}

        \begin{abstract}  
        	Cloud computing providers face the problem of matching heterogeneous customer workloads to resources that will serve them.  This is particularly challenging if customers, who are already running a job on a cluster, scale their resource usage up and down over time.
        	The provider therefore has to continuously decide whether she can add additional workloads to a given cluster or if doing so would impact existing workloads' ability to scale. Currently, this is often done using simple threshold policies to reserve large parts of each cluster, which leads to low efficiency (i.e., low average utilization of the cluster). We propose more sophisticated policies for controlling admission to a cluster and demonstrate that they significantly increase cluster utilization. We first introduce the cluster admission problem and formalize it as a constrained Partially Observable Markov Decision Process (POMDP).
        	As it is infeasible to solve the POMDP optimally, we then systematically design admission policies that estimate moments of each workload's distribution of future resource usage.
        	Via extensive simulations grounded in a trace from Microsoft Azure, we show that our admission policies lead to a substantial improvement over the simple threshold policy. We then show that substantial further gains are possible if high-quality information is available about arriving workloads.  Based on this, we propose an information elicitation approach to incentivize users to provide this information and simulate its effects.       
        	
        \end{abstract}

\maketitle


        \section{Introduction}
       
        Cloud computing is a fast expanding market with high competition where small efficiency gains translate to multi-billion dollar profits.\footnote{https://www.microsoft.com/en-us/Investor/earnings/FY-2018-Q2/press-release-webcast} 
Like many other markets (e.g., ridesharing platforms, kidney exchanges, online labor markets, and display advertising),         the efficiency of this market relies on the performance of a matching algorithm~\cite{ashlagi2019edge,ma2019tight,behnezhad2018almost,assadi2017stochastic}.
In the cloud computing case, the matching algorithm matches incoming requests for virtual machines to the hardware that will be used to satisfy them.

Despite the importance of this matching, most cloud clusters currently run at low efficiency. In the cloud domain, low efficiency means low \emph{average utilization}  of the cluster (i.e., only a relatively small fraction of resources are actually used by customers at any given time). There are many reasons for this ~\cite{YanSpark}. These include technical limitations (such as the need to reserve capacity for node failures or maintenance),  inefficiencies in scheduling procedures (especially if virtual machines (VMs) might change size or do not use all of their requested capacity), as well as factors that are external to the cluster (such as fluctuations in overall demand). Another important cause is the nature of many modern workloads: highly connected tasks running on different VMs that should be run on one cluster to minimize latency and bandwidth use~\cite{Cortez}. In practice, this means that different VMs from one user are bundled together into a \emph{deployment} of interdependent workload. When the workload of a user changes, his deployment can request a \emph{scale out} in the form of additional VMs or shut some of its active VMs down.

In this paper, we pay special attention to these size changes. 
Changing deployment sizes mean that providers face the difficult problem of deciding to which cluster to assign a deployment, as a deployment which is small today may, without warning, see a dramatic increase in size that must be accommodated.
To get a sense of the difficulty of this problem, consider that, over time, the number of VMs needed by a specific deployment could vary by a factor 10 or even 100, and a request to scale out should almost always be accepted on the same cluster, as denying it would impair the quality of the service, possibly alienating customers. Furthermore, once a deployment is running in one cluster, it would be error prone to move it to a different cluster, making such migrations only feasible as a last resort (e.g., because of hardware failure). 
Providers consequently hold large parts of each cluster as idle reserves to guarantee that only a very low percentage of these scale out requests is ever denied, leading to relatively low average utilization.
        
        \subsection{Cluster Admission Control}
        
 We reduce the problem of determining to which cluster to assign a new deployment to the problem of determining, for a particular cluster, whether it is safe to \emph{admit a deployment}, or if doing so would risk running out of capacity if some deployments scale~\cite{Cortez}. While a lot of research has been done on scheduling {\em inside} the cluster~\cite{omega,borg,tetrisched,zhao:cluster}, the admission problem has not been well studied before. Consequently, cloud providers are often still using simple policies like rejecting all new deployments once a cluster passes a fixed utilization threshold, effectively reserving a percentage of the cluster only for scale-outs.
 These may seem reasonable at first glance, as the law of large numbers might seem to suggest that with many jobs in a large cluster the current utilization would be a good guide to future utilization. But as \citeauthor{Cortez} (\citeyear{Cortez}) have shown, a relatively small number of deployments account for most of the utilization. This suggests that the \emph{types of deployments} (i.e., small/large, fast/slow scaling, short/long lived etc.) currently in a cluster have a larger impact on the failure probability than is apparent, and policies that only take the current utilization into account are suboptimal.
        
        
        \subsection{Overview of Contributions}

 We formalize the cluster admission problem as a constrained Partially Observable Markov Decision Process (POMDP) \cite{smallwood1} where each deployment behaves according to some stochastic process and the cluster tries to maximize the number of active compute cores without exceeding its capacity. Since the exact stochastic processes of individual arriving deployments are not known to the cluster, it has to reason about the observed behavior. The large scale of the problem as well as the highly complicated underlying stochastic processes make finding optimal policies infeasible, even for the underlying (fully observable) Markov Decision Process and with limited look-ahead horizon. 

 
        Since optimally solving this POMDP is not feasible, we next propose a strategy for constructing heuristic policies via a series of simplifying assumptions. These assumptions reduce the highly branching look-ahead space down to the approximation of a random variable using its moments. We then present the currently used threshold policy that does not take probabilistic information into account as well as two new policies that take successively higher moments into account.
        We fit our model to data from a real-world cloud computing center (Microsoft Azure internal jobs~\cite{Cortez}) and, via simulations, show that our higher moment policies produce a $30\%$ improvement over current practice, which would translate to hundreds of millions of dollars a year in savings for large cloud providers.

        In our basic model, relatively little is known about arriving deployments, so the performance gains we observe from our more sophisticated policies are driven by being able to better condition admission decisions on the current state of the cluster.  We next  examine how the utilization of the cluster can be further increased if more precise \emph{prior information} about arriving deployments is available. Prior work has explored similar opportunities in the context of resource planning and scheduling in
                analytics clusters~\cite{morpheus,perforator}.
  To study the value of prior information, we introduce a simple framework which captures a notion of the quality of information available. Through additional simulations, we quantify how our policies benefit from this additional information.  Depending on the quality of information available, the resulting gains increase to $50\%-65\%$ relative to current practice.

Finally, given the importance of the quality of this information, we design a new information elicitation mechanism, with the goal of simultaneously improving the cluster's utilization as well as the customers' utility.   This requires care to find a design that allows meaningful information to be elicited in an incentive compatible way while being simple for customers to use.  To this end, we propose that rather than explicitly asking users to describe the behavior of their deployments, cloud providers instead provide them with the opportunity to group their deployments into customer-defined categories with similar characteristics.  The cloud provider can then set a small portion of the fee for a deployment using a pricing rule based on the variance of resource demands of deployments in a category. We show that such variance-based pricing provides users with the right incentives to (a) label their deployments properly (into, e.g., high and low variance deployments) and (b) structure their workloads in a way that helps the cluster run more efficiently.  We provide additional simulations to quantify the benefits of an accurate labeling.
        
In practice, the magnitude of the gains from our approach will depend on many details our simulations elide.  However, we believe that our simulations provide a persuasive case that (a) there are substantial economic gains available from using our new admission policies for the process of matching deployments to clusters and (b) there are substantial further improvements possible by using our information elicitation approach to elicit relevant information from customers.

 While our work focuses on a specific problem faced by cloud providers, our overall approach is fundamentally about managing the tail risks of a stochastic process. In our case, these are the rare events where the cluster runs out of capacity.  Thus, our approach may also be of interest in other domains where the management of tail risks is important, for example in finance.
        
        \subsection{Related work}
                 There is a large literature on {\em cluster scheduling and load balancing} ~\cite{omega,borg,tetrisched,Wolke}.
                   In addition, some work  addresses a different notion of admission control to a cluster, namely how to manage queues for workloads which will ultimately be deployed to that cluster~\cite{delimitrou2013qos}.
In our work, however, while studying which deployments to admit to a cluster, we abstract away from the question of exactly which resources should be used, so our research is orthogonal to this prior work on scheduling and load balancing.
                
                        There is also a literature that views scheduling through the lens of stochastic online bin packing~\cite{cohen2016,song2014}.  This literature also deals with issues of changing workloads on possibly overcommitted resources.  However, the models in these papers operate at smaller scales and shorter time horizons. At these scales, the key phenomena we study are not present.           
                        
                     One of the core characteristics of the cluster admission problem is that the arrival of a new deployment into the cluster causes an increase in scale-outs in the future (i.e., until the deployment dies). This effect of ``work creating more work'' is also broadly reminiscent of mutually-exciting (Hawkes) processes (see  \cite{hawkes2018hawkes} for a recent survey).  These have also been studied in the context of queuing systems \cite{daw2018queues}, though there are some important differences (e.g., in our problem only accepted deployments cause scale outs and the rate of additional scale-outs varies greatly for different types of deployments).

                There is a large literature on market design challenges in the context of the cloud~\cite{cloudarticle}.
                Existing work has studied both queueing models where decisions are made online with no consideration of the future~\cite{fixedmarket_netecon,dierks2016cloud}
                and reservation models which assume very strong information about the future~\cite{azar2015truthful,babaioff2017era}. 
                Our work sits in an interesting intermediate position where users may have rough information about the types of their deployments.
               Furthermore, this literature focuses on using prices to determine which jobs should be served and which should not. While our problem is similarly about accepting or rejecting deployments, we do not want to ration through price discrimination. This is because (nearly) every request is ultimately served by the cloud computing provider and whether it goes into this or another cluster is of little consequence for the user.

Other market design work has looked at how multidimensional resources can be fairly divided among deployments. For example, \emph{Dominant Resource Fairness}~\cite{ghodsi:dominant} is an approach that has proven useful in practice~\cite{hindman:mesos} and has inspired follow-up work more broadly in the literature on fair division~\cite{parkes:beyond,dolev2012no,gutman2012fair,dynamicfairdivision_journal}. In our work, we assume that compute cores are the resource bottleneck and we do not model multi-dimensional resource requirements. Therefore, the considerations studied in the above papers are not present in our work.

        Solving POMDPs is a well-studied problem \cite{smith,russell,roy}. Unfortunately, finding an optimal policy is known to be PSPACE-complete even for finite-horizon problems \cite{Papadimitriou}. Even finding $\epsilon$-optimal policies is $NP$-hard for any fixed $\epsilon$ \cite{lusena}. In our case, the problem is further exacerbated by the existence of side constraints. Constrained POMDPS are far less well studied than unconstrained POMDPS. General (approximation) strategies proposed in the past include linear programming \cite{poupart,walraven}, point-based value iteration \cite{kim}, a mix of online-look ahead and offline risk evaluation \cite{undurti}, and forward search with pruning \cite{santana}. None of these approaches is efficiently applicable when the state space of the underlying MDP is large or, as in our case, partly continuous. \citeauthor{Khonji} (\citeyear{Khonji}) recently proposed a  fully polynomial time approximation scheme (FPTAS) for constant horizon constrained POMDPs. While their algorithm is polynomial in the size of the observation and action spaces, it is  exponential in the number of time steps. This makes it not applicable in domains with long time horizons like cluster admission control.
        While some work has addressed continuous state space POMDPs \cite{porta,duff,brooks}, none of this prior work is directly applicable to a constrained problem of the size we study in this paper. 

        \section{Preliminaries}
        In this section, we formally model the cluster admission problem and then introduce a POMDP formulation to solve the provider's control problem. 
        \subsection{ Formal Model}
        
        We consider a single cluster in a cloud computing center.
        A cluster consists of $c$ {\em cores} that are available to perform work, also called the cluster's \emph{capacity}.  These cores are used by {\em deployments}, i.e., interdependent workloads that use one or more cores.  The set of deployments currently in the cluster is denoted by $X$, and
        each deployment $x\in X$ is assigned a number of cores $C^x$. Any core that is assigned to a deployment is called \emph{active}, while the remainder are called {\em inactive}.\footnote{We assume that inactive cores can become active at any time. This means that features such as hardware failure or capacity reserved for maintenance are not modeled. This is a reasonable simplification, as they do not significantly affect the relative utilization of policies.}
        We do not model the exact placement of cores inside the cluster and in consequence we also do not model the grouping of cores into VMs. 
        
        A deployment can request to {\em scale out}, i.e., increase its number of active cores. Each such request is for one or more additional cores and must be accepted whenever activating the requested number of cores would not make the cluster run over capacity. Following current practice, scale out requests must be granted entirely or not at all.  Deployments may shut down some of their cores over time and these cores then become inactive.  A deployment {\em dies} when its number of active cores becomes zero. This can happen in two ways. First, it can die by successively shutting down one core after another until reaching zero active cores. Second, it can die spontaneously by shutting down all of its cores at once; intuitively this models a decision by a user to kill the deployment.\footnote{We model death as permanent because with no active cores any future request could be assigned to a different cluster.}

        A deployment $x$ is described by 4 deployment parameters $(C^x, \mu_x,\lambda_x, \sigma_x)$. We have already introduced the size of a deployment $C^x$. The remaining three parameters are drawn independently from population-wide distributions with PDFs $f_\lambda, f_\mu, f_\sigma$.
        We now explain how these parameters govern the behavior of the deployment.
                      
        At a high level, we assume that the deployments are memoryless (i.e., the basic processes governing a deployment's behavior only depend on the current state, which results in all processes following Poisson/exponential distributions). This is common in the literature whenever arrival and departure processes are modelled (e.g., in queuing theory), and has been used in previous models of cloud computing~\cite{fixedmarket_netecon,dierks2016cloud}. Memorylessness is reasonable at cloud scale and simplifies some calculations, but it is not essential for our approach and policies.

       Specifically, we assume that each core's \emph{lifetime} is distributed according to an exponential distribution with parameter $\mu_x$. The \emph{maximum lifetime} of a deployment (i.e., the time between arrival and it spontaneously shutting down all of its remaining cores) is distributed according to an exponential distribution with parameter $\Delta \mu_x$, where $\Delta$ is a (population-wide) multiplicative factor. This  effectively leads to an average maximum lifetime for the deployment of $\frac{1}{\Delta}$ average core lifetimes.                         
                The \emph{number of scale outs} per time unit for the deployment $x$ is distributed according to a Poisson distribution with rate parameter $\lambda_x \mu_x^\nu$, where $\nu$ is a population-wide parameter. This form of the rate parameter captures the empirical fact that
                 deployments with longer-lived cores scale slower than those with short lived cores. 
                The \emph{size of a scale out} is distributed according to one plus a Poisson distribution with parameter $\sigma_x$. While this is an approximation on an individual level (VM sizes usually come in powers of $2$), it is reasonable at the level of a cluster.

        New deployment requests arrive over time and are accepted or rejected according to an \emph{admission policy}.\footnote{A rejected deployment is only rejected from this cluster, not from the cloud computing center as a whole. While outside of our model, in practice it then simply gets sent to the next cluster.}
        The policy must limit the admission of new deployments to ensure that the cluster is not forced to reject a higher percentage of scale out requests than is specified by an internal \emph{service level agreement} (SLA) $\tau$.\footnote{This is a cluster-level SLA and not a deployment-level SLA, as in practice, the probability of tail-events such as scale out failures cannot feasibly be evaluated for a single deployment.} If a scale out request cannot be accepted because the cluster is already at capacity, one failure for the purpose of meeting the SLA is logged.
        An optimal policy therefore maximizes the {\em utilization} of the cluster, i.e., the average number of active cores, while making sure the SLA is observed in expectation.

        \subsection{The Provider's Control Problem: POMDP Formulation}
 The problem the provider is facing when deciding whether to admit a deployment is that the decision must be made under uncertainty regarding future arrivals and the future behavior of deployments. In addition, the provider cannot directly observe the parameters of each deployment's processes. 
 To understand how a provider can find a policy given this uncertainty, we model the problem as a Partially Observable Markov Decision Process (POMDP)$(\mathbb{S}, \mathbb{A}, \mathbb{R}, \mathbb{T}, \Omega , \mathbb{O})$ whose policy is constrained to meet the SLA $\tau$.
 
 
        For the POMDP formulation, we assume that time is discrete\footnote{While deployments can arrive at arbitrary times, it takes time to make the acceptance decision. Thus, there is little loss in discretizing time.} and that the problem has a finite time horizon\footnote{Our approach works for any   choice of horizon (or even an infinite horizon with average or discounted rewards).} denoted N. 
        The state space, denoted $\mathbb{S}$, describes the space of all possible states of the cluster. A state $s\in \mathbb{S}$ contains all information about the cluster's active deployments $X(s)$ (including, for each deployment $x$ both its current size $C^x$ and its scaling process parameters $\lambda_x, \mu_x$ and $\sigma_x$) as well as the deployments that arrived during the current time step. The action set $\mathbb{A}$ consists of individually accepting or rejecting each of the deployments that arrived this time step. The reward function $\mathbb{R}(s) =  \sum_{x\in X(s)} C^{x} $ is the number of active cores in a state $s$. The transition probability function is denoted $\mathbb{T}(s'| s,a)\forall s' s \in \mathbb{S}, \forall a \in \mathbb{A}$. Given a state of the cluster and admission decisions, this function captures the distribution over scale outs, core deaths, and arrivals of new deployments  that occur during the next time step.      
 $\Omega$ is the set of possible observations and $\mathbb{O}: \Omega \times \mathbb{S} \rightarrow \left[0,1\right]$ an observation model. 
  In our case, the observation model $\mathbb{O}$ is deterministic, but many states share the same observation.
 For state $s$, we always observe $\omega \in \Omega$ equal to the sizes of all deployments that are in the cluster and those that arrived with the last state transition.

As is standard, we further denote the cluster's current knowledge about which state $s$ it is in via a belief state $b\in \mathbb{B}$, i.e., a probability distribution over states. Specifically, a belief state $b$ specifies, for each deployment $x$ that is in the cluster or arrived with the last state transition, its current size  $C^x$ and (posterior) distributions $f_\lambda^x, f_\mu^x, f_\sigma^x$ over its scaling process parameters. For a given $x$, we let $\tilde{x}= (C^x,f_\lambda^x, f_\mu^x, f_\sigma^x)$, i.e., the provider's belief about the deployment $x$.  
A policy $\pi$ can now be defined as a mapping from belief states to actions. 

Whenever the cluster obtains a new observation $\omega \in \Omega$ in time step $n+1$, the belief state is updated according to the observation and transition models, i.e.,
                \begin{eqnarray}
                b_{n+1}(s'| b_n, a, \omega) \propto \mathbb{O}(\omega|s') \sum_s \mathbb{T}(s'|s,a)b_n(s).
                \end{eqnarray}
 Given this, we can now define two auxiliary functions. We let $g_{n,\pi, b}$ denote  the probability density function of the distribution over the states $s_n$ and belief states $b_n$ of the system $n$ time steps in the future, given policy $\pi$ and starting belief $b$. Furthermore, we let $h(s_n,\pi(b_n))$ denote the expected percentage of scale-outs that fail with the next state transition from a given state-action pair. We can now formulate the provider's control problem as finding an optimal policy given an SLA.
\begin{problem}[Cluster Admission Problem]
\label{Prob:ClusterAdmissionProblem}
The cluster admission problem is to find an optimal policy $\pi$ for the POMDP   $(\mathbb{S}, \mathbb{A}, \mathbb{R}, \mathbb{T}, \Omega , \mathbb{O})$ subject to the following two constraints:
\begin{align}
\int_{(s_n,b_n)} g_{n,\pi, b}(s_n,b_n) h(s_n,\pi(b_n))d(s_n,b_n)  \leq \tau &\quad  \forall \text{ safe } b  \; \forall 0\leq n< N  \label{const} \\
\pi(b) = \text{reject all arrivals}&\quad  \forall \text{ unsafe } b
\end{align}
        
         Here, we call belief state $b$ \emph{safe} if the policy $\pi_0$ which always rejects newly arriving deployments satisfies 
        \begin{align}
        \int_{(s_n,b_n)} g_{n,\pi_0, b}(s_n,b_n) h(s_n,\pi_0(b_n))d(s_n,b_n)  \leq \tau \quad \forall 0\leq n< N  \label{eq:unsafe}
        \end{align}
        and \emph{unsafe} otherwise.
                \end{problem}
                
                Intuitively, we would like our SLA constraint \eqref{const} to hold in every belief state. However, even if we follow an optimal policy,  we can reach belief states where (in retrospect) too many deployments have been admitted, such that, even if no new deployments are admitted ever again, the constraint \eqref{const} would be violated. Thus, if we would require Equation \eqref{const} to hold in all belief states, we would have an infeasible problem. 
         To address this, we do not enforce Equation \eqref{const} in \textit{unsafe} belief states (as defined in Problem~\ref{Prob:ClusterAdmissionProblem}). We instead require the policy to reject all arriving deployments until it reaches a \textit{safe} belief state.\footnote{The requirement to reject \emph{all} deployments is a design decision we revisit in Section \ref{learn}.}
        
                Note that the current time step is not referenced in Equation \eqref{const} or \eqref{eq:unsafe}. This is intentional to avoid horizon effects: a cluster should not aggressively start to accept new deployments close to the end of its lifetime.

                \section{A Tractable Problem Formulation}
                Optimal policies for the cluster admission problem (i.e., Problem~\ref{Prob:ClusterAdmissionProblem}) cannot be calculated in practice for three reasons. 
                First, there is no simple closed form for the state transition probabilities.  
                Second, the state space of the the POMDP is very large: consider a cluster with   20,000 cores. It usually has hundreds of deployments, each described by 4 parameters, some of which are continuous. Even discretized, this results in a state space exponential in the size of the cluster.   
                Third, even disregarding unlikely state transitions, the branching factor is large. This renders standard methods that rely on optimizing limited lookaheads infeasible. 
                Therefore, we now present three carefully chosen simplifying assumptions under which we characterize an optimal policy. In Section \ref{sec:heur}, we use this characterization to design practical admission control policies.

        \begin{assumption}[No Future Arrivals]\label{fut_arr} No further deployments arrive after the current timestep. 
        \end{assumption}
                
                 This assumption ensures that a policy does not reject deployments simply because better behaved deployments might arrive in the future. In the cloud domain, this behavior is desirable, as  even customers with high demand variability must be served by some cluster in the data center.

        \begin{assumption}[Relaxed Capacity Constraints]\label{as_rel} Deployments can scale out even if doing so exceeds the cluster capacity $c$. For the purpose of defining $h$, a scale out is considered to fail only if the cluster has already exceeded its capacity.  
        \end{assumption}    
 With no future arrivals, the cluster's future state only depends on how the sizes of the currently active deployments change.  However, if a cluster is full, further scale out requests by deployments are denied, introducing correlations between the future sizes of different deployments. Assumption \ref{as_rel} removes this correlation.  
  In particular, let $L^x_n$ denote the random variable that is the number of active cores of deployment $x$ in time step $n$. With the first two assumptions, $L_n^x$ is independent of $L_{n'}^{x'}$ for all $x \not = x'$ and all $n, n'$. The same holds for the random variable $L^{\tilde{x}}_n$ for the provider's belief. 
  This is reasonable because the cluster being full should be rare if the SLA is being met.
 
        \begin{assumption}[At Most one Event per Timestep]\label{one_event} In any timestep, at most one event occurs (i.e., at most one deployment scales out, shuts down cores, or arrives to the cluster).
        \end{assumption}

                Since the probability that more than one event occurs in a single time step approaches zero with increased granularity of the time discretization, it is reasonable to assume this.\bigskip 
                
                 Using these three assumptions, we can now simplify the problem of determining when the SLA constraint is met. Recall that $\tilde{x}=(C^x,f_\lambda^x, f_\mu^x, f_\sigma^x)$ specifies the provider's belief over a deployment $x$. In the following, we denote by $A_\pi(b)$ the set of beliefs $\tilde{x}$ over the active deployments in belief state $b$ and the deployments that are accepted with policy $\pi$ in belief state $b$.
                 
                \begin{proposition} \label{prop:main} For all policies $\pi$, under Assumptions \ref{fut_arr}, \ref{as_rel} and \ref{one_event}, the following holds:
        \begin{eqnarray}\label{simpBound}
        \int_{(s_n,b_n)} g_{n,\pi, b}(s_n,b_n) h(s_n,\pi(b_n))d(s_n,b_n)  =  Pr(\sum_{\tilde{x}\in A_\pi(b)} L_{n}^{\tilde{x}} > c) & \forall b  \; \forall 0\leq n< N . 
        \end{eqnarray}
                \end{proposition}
                \begin{proof}
                        To see the that Equation \eqref{simpBound} holds, note the following: 
                        By Assumption \ref{fut_arr}, it suffices to consider only the deployments that are currently in the cluster or arrive in the current time step, i.e., $\tilde{x}\in  A_\pi(b)$. 
                        By Assumption \ref{one_event}, at most one scale out can fail per time step. Thus the left hand side of Equation \eqref{simpBound} captures the following: if a scale out occurs in time step $n$, what is the probability that it fails. By Assumption \ref{as_rel}, a scale out fails exactly when $\sum_{\tilde{x}\in  A_\pi(b)} L_{n}^{\tilde{x}}>c$.                        
                \end{proof}                
                Using this result, it is straightforward to characterize an optimal policy for the simplified problem. 
                \begin{corollary}\label{cor:main}
                                Under Assumptions \ref{fut_arr}, \ref{as_rel} and \ref{one_event}, an optimal policy $\pi$ accepts an arriving deployment in belief state $b$ if and only if
                                \begin{eqnarray}\label{eq:corrEq}
                                Pr(\sum_{\tilde{x}\in  A_\pi(b)} L_{n}^{\tilde{x}} > c) \leq \tau &  \forall 0\leq n< N . 
                                \end{eqnarray}
                        \end{corollary}
                        \begin{proof}
                                By Assumption \ref{one_event}, it suffices to consider one arrival. 
                                                        By Assumption \ref{fut_arr}, if an arrival could be accepted without violating the constraint, doing so is optimal.
                                                        By Proposition \ref{prop:main}, the constraint is equal to Inequality \eqref{eq:corrEq}.                             \end{proof}
        Proposition \ref{prop:main} and Corollary \ref{cor:main} show that to implement an optimal policy for the simplified problem it suffices to evaluate the probability that a sum of independent random variables exceeds a threshold. The remaining question now is how to compute or approximate this probability for our complex processes fast enough to allow rapid responses to customer requests. 
         
        \section{Designing new Admission Control Policies}\label{sec:heur}
 In this section, we first define the complex random variables $L_n^{\tilde{x}}$ in terms of simpler random variables that directly arise from the processes. 
The essence of our approach is to take this definition of $L_n$ and use it to compute approximate moments of $L_n$ (i.e., approximate summary statistics of the behavior of the random variable). 
We then use these approximate moments to design new policies. 

We can describe $L_n^{\tilde{x}}$ using the following random variables (which have a superscript $\tilde{x}$ which we generally omit for brevity):
\begin{itemize}
        \item $C$ is the variable denoting the number of active cores at time step 0.
        \item  $Y_i$ is the random variable denoting the number of scale outs that occur between time step $i-1$ and time step $i$, assuming the deployment has not died.   
        \item $S_{i,l}$ is the size the $l$'th scale out request would have, assuming at least $l$ scale out requests occur between time step $i-1$ and time step $i$.  
        \item  $Z_{n,i,k}$ is the binary random variable denoting whether the $k$'th core activated between time steps  $i-1$ and $i$ would still be active in time step $n$, assuming at least $k$ cores were activated and the deployment has not died.  For $i=0$, this instead refers to whether the $k$'th core that is active at timestep $0$ is still active at time step $n$.
        \item $D_i$ is the random variable which is 1 if $x$ would not have died due to a lack of active cores before time step $i$.  It can be defined recursively as
        \begin{eqnarray}
                D_i &=& D_{i-1} (1- \Pi_{k=1}^{C}(1-Z_{i,0,k}) \Pi_{j=1}^{i-1} \Pi_{k=1}^{\sum_{l=0}^{Y_i}S_{i,l}} (1-Z_{i,j,k}))\\
                D_1 &=& 1-\Pi_{k=1}^{C}(1-Z_{1,0,k})
        \end{eqnarray}
        \item   $B_n$ is the random variable denoting the number of cores that were active at time step 0 and are still active in time step $n$, which can be calculated as
        \begin{eqnarray}
                B_n &=&\sum_{k=1}^{C}Z_{n,0,k}.
        \end{eqnarray}
        \item   $Q_n$ is the random variable denoting the number of cores activated between time step 0 and time step $n$ that are still active assuming no service termination, i.e.
        \begin{eqnarray}
                Q_n &=&\sum_{i=1}^{n} \sum_{k=1}^{\sum_{l=0}^{Y_i}S_{i,l}}Z_{n,i,k}.
        \end{eqnarray}
        \item $M_i$ is the random variable which is 1 if the maximum lifetime of the deployment is at least $i$ and 0 otherwise.
        \item Finally, $L_n$ can be calculated as $ L_n= M_n D_n ( Q_n + B_n)$.
        
\end{itemize}

We now turn to the design of our approximate policies. 
                \subsection{Baseline (Zeroth Moment Policy )}
                Before introducing our new policies, we state the  
                 baseline admission control policy that is widely used in practice. It is a myopic policy that simply compares the current number of active cores to a threshold. This policy does not use any information about the set of deployments besides the total number of active cores. It can be we viewed as a degenerate case of our approach, as it does not take any probabilistic information about the random variables into account. We therefore also call it a \emph{Zeroth Moment Policy}. Because it uses a limited amount of information, it must be conservative in how many deployments it accepts, since it does not know how often or fast they will scale out.
                
                \begin{definition}[Zeroth Moment Policy (Baseline)]
                        Under a zeroth moment policy $\pi$ with threshold $t$, a newly arriving deployment is accepted if, after accepting the deployment, there would be less than $t$ cores active.
                \end{definition}
        
        \subsection{First Moment Policy}
        Our first policy approximates the probability of scale out failures (i.e., Equation~\eqref{simpBound}) by utilizing the first moments, i.e. the expected value of the deployment processes. 
        By Markovs's Inequality, for a non-negative random variable $L$ and $c \geq 0$, it holds that
        \begin{eqnarray}
                Pr(L\geq c) &\leq& \frac{E[L]}{c}.
        \end{eqnarray}
 Such a policy that utilizes Markov's Inequality therefore rejects an arriving deployment when the expected utilization lies above a chosen threshold and otherwise accept.
        \begin{definition}[First Moment Policy]
                Under a first moment policy $\pi$ with threshold $t$, a newly arriving deployment in belief state $b$ is accepted if, after accepting the deployment, the expected number of active cores would be less than $t$ in all future time steps, i.e.
                \begin{eqnarray}
                \sum_{\tilde{x} \in  A_\pi(b)} E[L^{\tilde{x}}_n] \leq t \quad \forall 0\leq n< N, 
                \end{eqnarray}
                where $E[L^{\tilde{x}}_n]$ is approximated, for example using the approach described in Proposition \ref{prop:expap}. 
        \end{definition}

\begin{proposition}\label{prop:expap}
        Assuming all $M_n$, $D_i$, $Q_n$, and $B_n$ are uncorrelated and $Z_{i,j,k}$, $Y_i$, and $S_{i,l}$ are uncorrelated as constituents of $D_i$, it holds:
        \small
        \begin{eqnarray}
        E[L_n] &= & E[M_n]E[D_n] ( E[Q_n] + E[B_n])\\
        E[Q_n]&=& \sum_{i=1}^{n} E[E[Y_1|\lambda, \mu]E[S_{1,1}| \sigma]E[Z_{n,i,1}|\mu]]\\
        E[B_n] &=& C E[Z_{n,i,k}]\\
        E[D_i] &\leq& E[D_{i-1}]( 1 -  (1-E[Z_{i,0,1}])^{(C)} \Pi_{j=1}^{i-1} (1-E[Z_{i,j,1}])^{E[Y_1]E[S_{1,1}]})\\
        E[D_1] &=& (1-(1-E[Z_{1,0,1}])^{(C)})
        \end{eqnarray}
        \normalsize
\end{proposition}
         The proof is provided in Appendix~\ref{apx:OmittedProofs}. It works by direct calculation and applying Jensen's Inequality to $D_i$. While ignoring some correlations introduces a nontrivial error into the approximation, this is done to ensure that the expectation can be evaluated in linear time. Additionally, the tail bounds from Markov's Inequality are relatively loose.
This makes it necessary to calibrate $t$, as simply setting $t = \tau c$ would yield excessively conservative policies.  Nevertheless,
        as we will see in Section \ref{sim1}, this approximation still carries enough information for our policies to work well once tuned.\footnote{Similar observations have been made in the literature on the use of effective bandwidth for admission control in queueing settings~\cite{kelly1991effective,berger1998extending}.}

        \subsection{Second Moment Policy}
        First moment policies do not take much information about the structure of deployments into account. In a sense they have to always assume the worst possible population mix and run the risk of accepting deployments with low expected size but high variance when close to the threshold. One way around this is to also take the second moment, i.e., the variance of $L_n$, into account. To address this, we propose to use Cantelli's inequality, a single-tailed generalization of Chebyshev's inequality, to approximate the probability of scale out failures (i.e., Equality~\eqref{simpBound}). Cantelli's inequality states that, for a real-valued random variable $L$ and $\epsilon\geq 0$, it holds that
                \begin{eqnarray}
                Pr(L-E[L]\geq \epsilon) &\leq& \frac{Var[L]}{Var[L]+\epsilon^2}.
                \end{eqnarray}

        If we now set $\epsilon = (c-\sum_{x \in X} E[L^x_n])$, we obtain a bound for the probability of running over capacity that takes more information into account than a first moment policy.
        \begin{definition}[Second Moment Policy]
                Under a  second moment policy  $\pi$ with threshold $\rho$, a newly arriving deployment in belief state $b$ is accepted if, after accepting the deployment, the estimated probability of running over capacity would be less than $\rho$ in all further time steps, i.e.
                \begin{align}\label{second_eq}
                \sum_{\tilde{x} \in  A_\pi(b)} E[L^{\tilde{x}}_n] \leq &c & \quad \forall 0\leq n< N \\
                \frac{\sum_{\tilde{x} \in  A_\pi(b)} Var[L^{\tilde{x}}_n]}{\sum_{\tilde{x} \in A_\pi(b)} Var[L^{\tilde{x}}_n]+(c-\sum_{\tilde{x} \in A_\pi(b)} E[L^{\tilde{x}}_n])^2} \leq& \rho& \quad \forall 0\leq n< N 
                \end{align}
                where $E[L^{\tilde{x}}_n]$ is approximated using the approach described in Proposition \ref{prop:expap} and $Var[L^{\tilde{x}}_n]$ is approximated using the approach described in Proposition \ref{VARFORM}.
        \end{definition}
        
        \begin{proposition}\label{VARFORM}
                Assuming all $M_n$, $D_i$, $Q_n$, and $B_n$ are uncorrelated, it holds:                
                \begin{eqnarray}
                V[L_n] &=&  E[M_n]^2 V[D_n (Q_n+B_n)] + V[M_n]E[D_n (Q_n+B_n)]^2\\
                &&+V[M_n]V[D_n (Q_n+B_n)]\\
                V[D_n (Q_n+B_n)] &=& E[D_n]^2 (V[Q_n]+V[B_n]) + (E[Q_n]+E[B_n])^2 V[D_n]\\
                &&+ V[D_n] (V[Q_n]+V[B_n])\\
                V[Q_n]&=& E[V[Q_n|\lambda,\sigma,\mu]] +V[E[Q_n|\lambda,\sigma,\mu]]\\  
                V[Q_n|\lambda,\sigma,\mu] &=&\sum_{i=1}^{n} \left(  \left( V[Y_i|\lambda,\mu] E[S_{i,l}|\sigma]^2 \right. \left. + E[Y_i|\lambda,\mu] V[S_{i,l}|\sigma]\right) E[Z_{n,i,1}|\mu]^2 \right.\\
                &&\left. +E[Y_i|\lambda,\mu]E[S_{i,l}|\sigma] V[Z_{n,i,1}|\mu] \right)\\
                V[B_n] &=&CE[  V[ Z_{n,i,k}|\mu]] +C^2 V[E[Z_{n,i,k}|\mu]] \\
                V[D_n] &=& E[D_n]-E[D_n]^2
                \end{eqnarray}
                \normalsize
        \end{proposition}
        The proof is provided in Appendix~\ref{apx:OmittedProofs}. It works by direct calculation. 
 Note that, since the expectation is approximated as given in Proposition \ref{prop:expap}, $V[D_n]$ carries over any approximation errors from $E[D_n]$. 
 As with first moment policies, the bound given by the inequality is again not tight enough to simply set it to $\rho =\tau$ and $\rho$ has to be tuned. 
        \subsubsection{Computational overhead.}
        The computational overhead of the second moment policy depends on the number of future time steps it evaluates and the chosen prior distributions for the provider's belief state. As long as well-behaved priors are used (e.g., the Gamma priors we use in our simulations),  each single rule application is fast.  For such priors, updating the estimate for the second moment policy for a single deployment can be done in $O(n)$ where $n$ is the number of evaluated time steps. Whenever a new deployment arrives, the estimate is updated for every active deployment. This leads to a worst case runtime of $O(|X| n)$ where $|X|\leq c$ is the number of active deployments. For multiple clusters this is fully parallelizable at the cluster level because each cluster has its own policy evaluation. Updating the prior of a deployment during runtime has negligible complexity ($O(1)$). A cloud computing center consisting of clusters of capacity $c$ with an arrival rate of $L$ new deployment requests per hour therefore has a computation overhead of at most $O(L c n)$  each hour, parallelizable into jobs of size $O(n)$.\footnote{If further ML is (optionally) employed to obtain an individual prior for arriving deployments (as discussed in Section \ref{learn}), that computation time would need to be added and depends on the algorithm in question.} This means that even relatively large look-ahead horizons $n$ can easily be implemented in practice.  

        \section{Empirical Evaluation}\label{sim1}

 In this section, we evaluate the performance of our admission policies using a model fitted to the real-world data trace of \citeauthor{Cortez} (\citeyear{Cortez}).

 \subsection{Data Trace and Fitted Model}
 \citeauthor{Cortez} (\citeyear{Cortez}) published a data trace consisting of all deployments that populated a Microsoft Azure datacenter in one month.  
 Since the data set is of limited size and only covers one month, we cannot directly evaluate the policies on the historical deployments. One month is too short to fully evaluate cluster admission policies as many effects only show up after months of usage. Instead, we fit processes to the data we do have, to simulate longer time periods (3 years, in our simulations). We defer evaluations against real
 deployments to future work. 
 
 An in-depth discussion of our fitting procedure can be found in Appendix \ref{fitting}.
 The resulting model utilizes Gamma priors, which are a very general distribution (containing the Chi-squared, Erlang and Exponential distributions as special cases) and fit the data well. The fitted parameters are shown in Table~\ref{table1}. The moment approximation resulting from combining Propositions \ref{prop:expap} and \ref{VARFORM} with these priors is given in Appendix \ref{A:GAMMA}.
 In the following we present the results of our simulations.


        \begin{table*}[t!]
                \small
                
                \centering
                
                \begin{tabular}{cc}
                                \hline
                        \multirow{3}{*}{Priors} & $\mu \approx Gamma(0.3107,0.5778)$ \\
        
                        & $\lambda \approx Gamma(0.4907,0.4496)$\\
                        
                        &$\sigma \approx Gamma(0.2616,0.0552)$\\
                                                \hline \noalign{\smallskip                              }
                                                
                                                \hline 
                        \multirow{2}{*}{Global Parameters} & $\Delta = 0.119$ \\
                        
                         &  $\nu = 0.673$   \\
                                \hline \noalign{\smallskip}                     
                \end{tabular}
                \caption{Fitted processes}      \label{table1}
        \end{table*}

 \subsection{Simulation Setup}

        We simulate clusters with capacity $c=20,000$ for a 3-year period with all three policies.         An average of $1$ new deployment per hour arrives according to a Poisson process. The parameters of each arriving deployment are drawn from the fitted distributions presented in Table \ref{table1}. We tune the threshold for each policy via binary search, subject to meeting an SLA of $0.01\%$.\footnote{This SLA is somewhat stricter than is typically used in practice, which helps counterbalance our model abstracting away complexities such as fragmentation and node failure.} 
        We verify that the SLA is satisfied on average (over runs and months).

Evaluating our first and second moment policies with a three year time horizon and fine-grained time steps is fast enough to be done in real time. However, doing so would take too much computation power to simulate the thousands of years of cluster operation required for our experiments. 
Therefore, we use the following approach to simulate clusters with a three-year lifespan with a reasonable number of core-hours. We divide the  first and second moment policies into 5 subpolicies and only accept a deployment if all subpolicies accept it. The subpolicies have increasingly fine-grained time steps, but each only evaluates a limited look-ahead horizon:  3 years, 1 year, 1 month, 1 week, and 24 hours. Each subpolicy discretizes its time into 600 timesteps.  We performed 500 runs and report the average utilization across all runs. Since failures to scale out are focused in the tail of the runs (e.g., with the tuned zeroth moment parameter only about $1\%$ of runs contain any failures), we employ importance sampling to obtain sufficient samples from the tail to guarantee SLA satisfaction with high confidence.  Details about the importance sampling can be found in Appendix \ref{AP:SAMP}. To avoid misestimating confidence intervals with biased data, we report $95\%$  bias-corrected and accelerated bootstrap confidence intervals (following \cite{efron}, $100000$ re-samples) instead of standard errors.

        \begin{table}[t]

                \centering
        
                        \vspace{2pt}
                \begin{tabular}{rcccc}
                        \noalign{\smallskip} \hline \noalign{\smallskip}
                        
                        \textbf{Policy} &      \textbf{Threshold} & \shortstack{ \textbf{Utilization}  }  \\
                        \hline \noalign{\smallskip}
                        Zeroth Moment          & $t=8,864$  & $50.45\%$ $ (48.2, 52.7)$  \\

                        \noalign{\smallskip} \hline \noalign{\smallskip}
                        \noalign{\smallskip} \hline \noalign{\smallskip}
                        
                        First Moment&     $t=14,223 $& $66.19\%$ $(63.41,68.94)$ \\

                        \noalign{\smallskip} \hline \noalign{\smallskip}
                        \noalign{\smallskip} \hline \noalign{\smallskip}
                        Second Moment &     $\rho=0.112 $&  $67.32\%$ $(64.35,70.26)$ \\
                        \hline \noalign{\smallskip
                        }
                \end{tabular}
                        \caption{Simulation results showing the performance of the three policies. $95\%$ bootstrap confidence intervals are shown in parentheses.}
                                \label{table2}
        \end{table}
        
        \subsection{Results}

         We now compare the utilization of our policies to the industry baseline zeroth moment policy. The results are summarized in Table \ref{table2}. The zeroth moment policy obtains its best result with a threshold of $t=8,864$, i.e., new deployments are accepted whenever less than $8,864$ would be active in case of acceptance. This results in an average utilization of $50.45\%$ over the lifetime of the cluster.  The first moment policy with threshold $t=14,223$ increases the utilization by $15.74$ percentage points to $66.19\%$. This constitutes a relative increase in utilization of $31.2\%$ over the zeroth moment policy. Similarly, the second moment policy with threshold $\rho=0.112$ achieves a utilization of $67.32\%$, a relative improvement of $33.44\%$. 
         
         At first sight, it may be surprising that the first and second moment policies achieve similar utilization. However, this can be explained as follows. Under both policies, the overwhelming number of simulated clusters never reject a scale out request. However, in a few runs, too many  large, long-lived deployments are accepted in the beginning of a cluster's lifetime. This leads to many rejections months or even years in the future. Since this happens early in a cluster's lifetime when not much is known about deployments, the difference between the first and second moment policies is relatively small. This highlights the value of obtaining additional (probabilistic) information about arriving deployments. We study this in the next section.

        \section{The Value of Deployment-specific Priors}
        \label{learn}

        So far, we have assumed that the cluster does not have any information about arriving deployments, except for the initial number of cores. The acceptance decision therefore had to primarily depend on the state of the deployments that are already in the cluster.

Intuitively, a policy could more precisely control whether accepting a deployment would risk violating the SLA  if the policy had more information about the future behavior of the specific deployment. One way to obtain such information would be to use machine learning (ML) based on features of the arriving deployment and past deployment patterns of the submitting user~\cite{Cortez}. While evaluating particular ML algorithms is beyond the scope of this paper, we evaluate the effect that different levels of available information have.
        To do this, we need to parameterize the level of knowledge. For this we assume that the cluster simply gets passed some number of observations from each true scaling process distribution of each arriving deployment.\footnote{As we have used conjugate prior distributions in our model, this approach matches the standard interpretation of parameters of the posterior distribution in terms of ``pseudo observations.''}

                \subsection{Improving the Handling of Short-lived Deployments}

                         Our moment policies as defined so far cannot yet make optimal use of this additional prior information. While an optimal policy with good prior information would balance the admission of long-lived and short-lived deployments to keep the utilization more stable over time, the moment policies always accept new deployments on a first-come first-served basis until their constraints are violated. This means that if many very long-lived, slow-scaling deployments arrive in the beginning, the cluster sometimes quickly reaches unsafe belief states in which it stops accepting any new deployments, but for which the critical event lies months or even years in the future.
While stopping the admission of new deployments in such a situation is reasonable when no prior information about arriving deployments is available, \emph{with} prior information the policy might know that some arriving deployments will almost surely be dead by the time the cluster has filled up.
         To make use of this, we now present a heuristic modification of our moment policies such that the resulting policy is allowed to accept deployments that only have a marginal impact on the possible SLA violation, even in unsafe states. As a simple condition for this, we call a deployment \emph{marginal in timestep $n$} if its expected size is smaller than $10^{-5}$, i.e.,  $E[L^x_n] < 10^{-5}$.

                \begin{definition}[Marginal Heuristic]
                        Under a first or second moment policy $\pi$ with the  \emph{marginal heuristic}, a newly arriving deployment $x$ in belief state $b$ is accepted if in each future time step $n<N$, after accepting the deployment, either the underlying moment policy's condition is satisfied or the arriving deployment $x$ is marginal, i.e., $E[L^x_n] < 10^{-5}$.
                \end{definition}
                
                Going forward, we use the marginal heuristic, unless explicitly noted. It should be pointed out that this heuristic does not have any effect when the cluster does not have good prior information about arriving deployments. With only the global prior, no deployment is marginal for any future timestep $n<N$.

                                \subsection{Simulation Results}

        In this section, we present simulation results to demonstrate the value of deployment-specific priors. 
        We simulate the first and second moment policies (with marginal heuristic), now with four different levels of prior information:\ 0, 1, 5, and 50 observations. Otherwise, we use the same simulation setup as in Section \ref{sim1}. The results are shown in Figure \ref{simref}.\footnote{We also simulated our policies without the marginal heuristic (see Appendix \ref{AP:AB}), and we observe the same general patterns. As one would expect, without the marginal heuristic, the achieved utilization is somewhat smaller (especially with good priors).} We see that having prior knowledge equivalent to even a single observation improves utilization significantly, resulting in a utilization of $75\%$ and $79.5\%$ for the first and second moment policies, respectively.
        Better priors lead to even better utilization, with the second moment policy reaching a utilization of $83.8\%$ with $50$ observations. 
        
        While it is infeasible to calculate the utilization corresponding to an optimal solution of the POMDP, we have derived an \textit{upper bound }of $92.1\%$ by analyzing policies that do not have to satisfy any SLA. Thus, the second moment policy with good prior information achieves more than $90\%$ of the theoretically achievable as given by this (unreachable) upper bound, while delivering a relative  increase in utilization of $24.48\%$ above the same policy without prior information and a $66.1\%$ increase over the baseline (i.e., zeroth moment) policy. This shows both the power of our policies and the great importance of taking all available prior information about arriving deployments into account.

        \begin{figure}[t]
                \centering%
                \includegraphics[height=8cm]{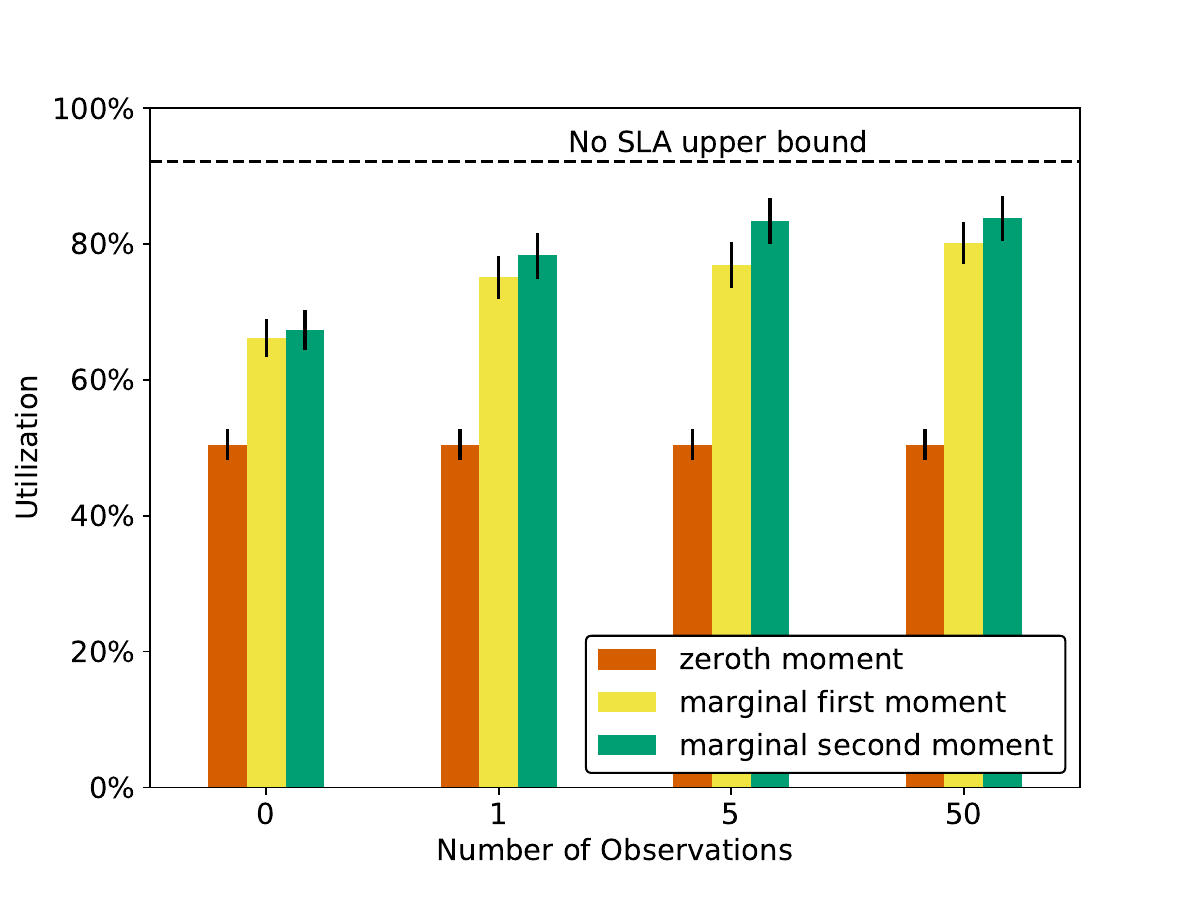}
                \caption{Performance of different policies depending on prior information (error bars indicate $95\%$ bootstrap confidence intervals)}
                \label{simref}
        \end{figure}%

        \section{An Elicitation Mechanism to Improve Priors}
        
        Given the importance of the quality of prior information that we established in the last section, in this section, we use techniques from mechanism design to improve this quality.  Our approach assumes that users do not typically submit deployments with arbitrary parameters. Instead, they may have a small number of different \emph{types} of deployments.  However, the typical mechanism design approach of using a direct revelation mechanism, where customers reveal their full type, seems problematic. First, it may be very cumbersome from a user interface perspective. Second, customers may not have such a detailed understanding of their deployments and thus would  run the risk of being penalized for a ``misreport.''   Instead, we seek a design that allows meaningful information to be elicited in an incentive compatible way while being simple for customers to use.  To this end, we propose that rather than explicitly asking users to describe the behavior of their deployments, cloud providers instead provide them with the opportunity to group them into customer-defined categories of roughly similar deployments.  Learning priors for each individual category then results in more precise priors and higher utilization. To incentivize such grouping, the cloud provider can set a small portion of the fee for a deployment using a pricing rule based on the variance of resource demands of deployments in a category. 
        We now first present and analyze such a variance-based pricing mechanism  and then evaluate the potential utilization gains this mechanism may produce via additional simulations.

        \subsection{Variance-based Payment Rule}
        Typically, users are charged a fixed payment per hour for each core their  deployment uses. With a variance-based payment rule, we add a small additional charge based on the variance of the estimate for the deployment's scaling process and allow users to \emph{label} the type of their deployments, resulting in an hourly \emph{variance-based payment rule} $q(x)$ of the form:
        \begin{eqnarray}
        q(x)= \kappa_1 C^x + \kappa_2 Var(x), \label{eq:payment}
        \end{eqnarray}
        where $\kappa_1$ and $\kappa_2$ are price constants and $Var(x)$ is an estimate of the variance of the deployment.
        A payment rule of this form incentivizes users to assign similar labels to similar deployments to minimize the estimated variance.
        
        
        To see this, consider a user who has two types of deployment, $x$ and $y$, with true variances $Var(x)$ and $Var(y)$. He could now simply submit the deployments under a single label. For the provider, this means that each submitted deployment is of either type with a certain probability, which increases the variance of her prediction. But if the user would label each deployment appropriately with either ``$x$'' or ``$y$,'' then the provider would know for each arriving deployment which type it is, reducing variance and therefore the need to reserve capacity. The following proposition, which is immediate from the law of total variance, shows that, at least in the long run, labeling his deployments also reduces a user's payments.

        \begin{proposition}\label{marketprop}
                Let $z$ be the mixture that results from submitting one of two types of deployments $x$, $y$ chosen by a Bernoulli random variable $\alpha \sim Bernoulli(p_\alpha)$, i.e., such that $z$ is of type $x$ with probability $p_\alpha$ and of type $y$ with probability $1-p_\alpha$. Then it holds that
                \begin{eqnarray}
                p_\alpha Var(x) + (1-p_\alpha) Var(y) \leq Var(z)
                \end{eqnarray}
        \end{proposition}
        \begin{proof}
                Since $z$  has finite variance, the law of total variance states:
                \begin{eqnarray}
                Var(z)&=&E [Var (z|\alpha)]+Var (E [z| \alpha])\\
                &\geq &E [Var (z|\alpha)]\\
                &=& p_\alpha Var(x) + (1-p_\alpha) Var(y)
                \end{eqnarray}                          
        \end{proof}
        Proposition \ref{marketprop} shows that the user would be better off by splitting the mixture and submitting the deployments under separate labels, directly resulting in the following corollary.

        \begin{corollary}
         Under any variance-based payment rule $q(x)$ of the form given in Equation~\eqref{eq:payment} with $\kappa_2>0$, it is a dominant strategy for users with multiple deployment types to label deployments by type. 
        \end{corollary}
         Note that this corollary abstracts away issues of learning and non-stationary strategic behavior; but for reasonable learning procedures we expect a consistent labeling to lead to lower variance than a mixture while learning. Further, this approach not only gives the user correct incentives to reveal the desired information, but actually incentivizes him to improve the performance of the system.  In particular, another way he can lower his payment under this scheme (outside the scope of our model) is to design his deployments in such a way that they have lower variance in their resource use.  Since more predictable deployments would allow the policy to maintain a smaller buffer, this provides an additional benefit to the system's utilization.
     
        \begin{remark}
       The per core-hour payments of a user can be based on \emph{a-priori} or \emph{a-posteriori} estimates of a deployment's variance. With an a-priori estimate, the user knows his payment (per core-hour) before he starts his deployment, which can be very important for certain   users. On the other hand, such a payment rule invites strategic deployment submissions: a user could   submit a number of small low-variance deployments before submitting a large high-variance deployment, with the goal of reducing the provider's estimate and thus his payment for the large deployment. The provider could mitigate the potential gain of such a manipulation by carefully choosing the estimation procedure, so it is unclear how frequent and successful such manipulations would be in practice.           
          With an a-posteriori estimate (i.e., the user's hourly payment is based on the variance estimate of the deployment at the end of the hour), such strategic deployment submissions could be made unprofitable. However, now users do not know their exact payments in advance. To cater to users that require a fixed price before submitting a deployment, the provider might want to set an upper limit on prices and advertise lower prices as  discounts.          
         Thus, which type of estimate is optimal for a given provider depends on the requirements of her user base.
        \end{remark}

        How much any given user could ultimately save by labeling his deployments mostly depends how different his deployment types are and on how high the provider sets the charge for variance. A user whose deployments are quite uniform will not save much, while a user with some deployments which never scale and some that scale a lot can potentially save a lot. Note that how much the provider should charge is not immediately clear. While she would want to set a high price to put a strong incentive on users, she also has to keep the competition from other providers in mind. At what point the loss of market share outweighs the gain in utilization is an intriguing problem we leave for future work.

                \subsection{Simulation Results}
                
                To illustrate the potential gains in utilization of such a variance-based payment rule, we consider a setting where all users have two deployment types, drawn independently from the same population distribution as in Section \ref{sim1}.  We assume that each user only submits a single deployment and then departs, but that the provider has $5$ prior observations each from every user's two deployment types. Otherwise, the simulation setup is again the same as in Section \ref{sim1}. In this setting, we contrast the utilization of a provider employing a  second moment policy with and without employing a variance-based payment rule.  With the variance-based payment rule (and users consequently declaring their type), the setting becomes equivalent to the one presented in Section \ref{learn} with $5$ observation. When the  variance-based payment rule is \textit{not} used (and users consequently do not label their deployments),  we assume that the provider updates her belief for both types independently and evaluates her second moment policy on the mixture. 

\begin{figure}[t]
                        \centering%
                        \includegraphics[height=8cm]{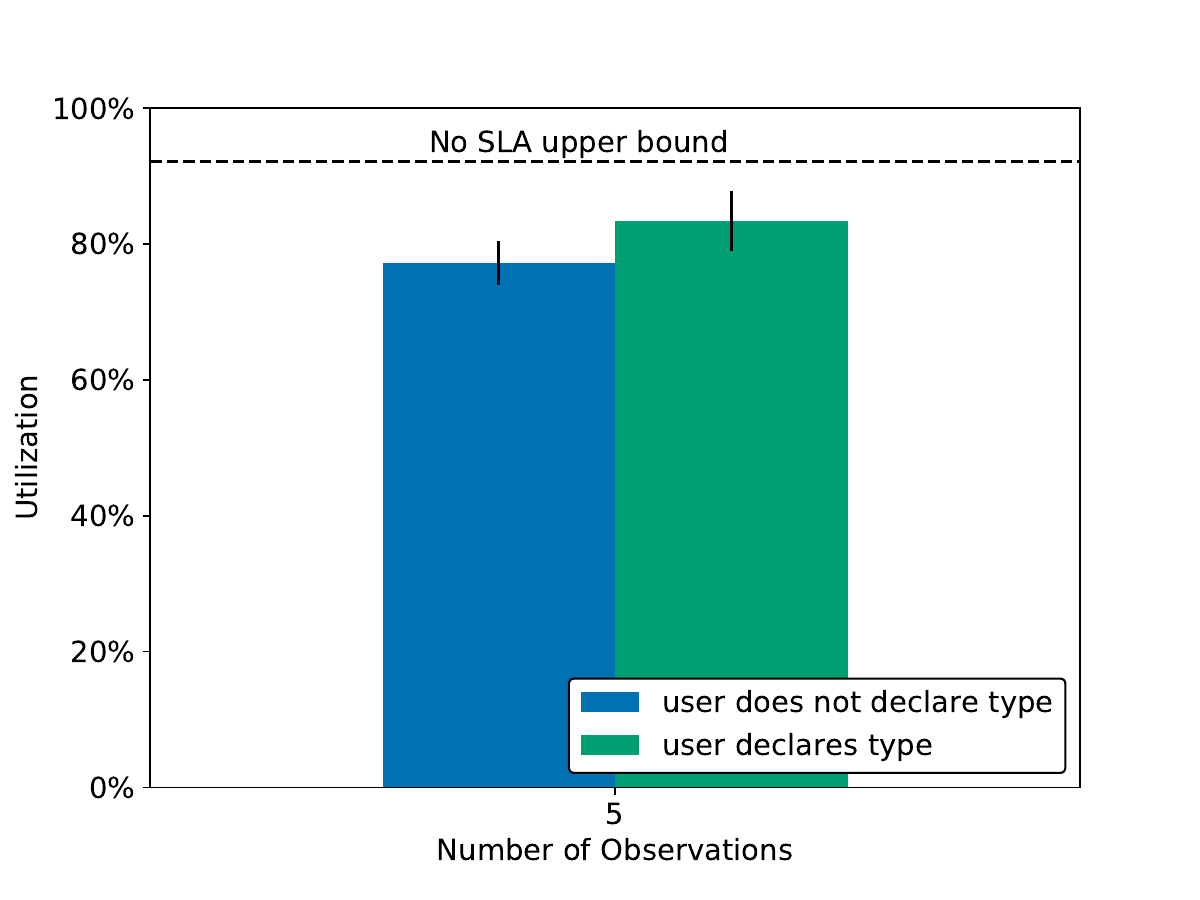}
                        \caption{Performance of the second moment policy with two deployment types per user and $5$ observations  (error bars indicate $95\%$ bootstrap confidence intervals)}
                        \label{fig:report}
                \end{figure}%

                As we can see in Figure \ref{fig:report}, when users do not label their deployments, this yields a utilization of $77\%$. In contrast, when users do label their deployments, then (as expected) the utilization increases to $83\%$. This shows that, from a cluster point of view, employing a variance-based payment rule leads to a sizable increase in utilization.  
\section{Conclusion}
We have studied the problem of cluster admission control for cloud computing, where accepting demand now causes unrejectable demand in the future. The optimal policy is given as the solution to a very large constrained POMDP, which is infeasible to solve. In practice, simple threshold policies are employed for admission control. In contrast, we have proposed multiple more sophisticated policies. Our results demonstrate that the utilization can be increased  by approximately $30\%$ just from learning about deployments while they are active in the cluster. 
Furthermore, we have shown  that this can  be improved to a $50-65\%$ gain if better prior information about arriving deployments is available, for example through learning or elicitation techniques. 
Even though the realized gains in practice are likely to be somewhat lower due to practical engineering constraints (e.g., the need to handle node outages), they should still be sizable.  
 At cloud scale, even savings of a few percent translate to many hundreds of millions of dollars, and any dollar saved directly translates to a gross profit increase for the cluster provider.

Our work points to a number of interesting future research directions. 
We have only looked at cluster admission policies at the level of a single cluster, abstracting away the question of which cluster should be chosen, implicitly assuming a first-fit or random-fit heuristic. Future research should look at the question whether filling all clusters with the same mixture of deployments is reasonable or if dedicating different clusters to different types of deployments could be used to further increase utilization. 
A related direction is that our model and policies assume that deployment behavior does not change during runtime. While this is a reasonable approximation for many deployments, some long running deployments might exhibit more involved life cycles in practice. One way for policies to account for this is to discount past observations. 

There are also open questions regarding mechanism design in the cloud domain. In subsequent work, \citeauthor{dierks2020competitive} (\citeyear{dierks2020competitive}) have already analyzed the competitive effects of employing the variance-based payment rule we proposed in this paper in a duopoly model. They find that, in equilibrium, while using a variance-based payment rule weakly increases welfare, the effects on the providers' profits are ambiguous. However, their analysis does not take into account that a variance-based payment rule can yield better priors about submitted deployments and thus improve the efficiency of the provider's admission policy. Future work could explore an alternative economic design, where the provider offers a menu of two alternatives to her users: the standard alternative, where deployments can always scale out; and a cheaper alternative, where deployments are not allowed to scale out (or can scale only with a ``best effort'' guarantee).  Such an approach could be viewed as implicitly selling finance-style options on the ability to scale out.

%

\appendix

\bibliographystyle{ACM-Reference-Format}
\bibliography{admission}

\newpage
\appendix



\section{Omitted Proofs}
\label{apx:OmittedProofs}

\begin{proof}[Proof of Proposition \ref{prop:expap}]
        \begin{itemize}
                
                                \item 
                                If $M_n$, $D_n$, $Q_n$ and $B_n$ are uncorrelated, it holds by linearity and multiplicativity of the expected value for uncorrelated random variables:
                                \begin{eqnarray}
                                E[L_n] &=& E[M_n]E[D_n] ( E[Q_n] + E[B_n])
                                \end{eqnarray}

                \item $Q_n$:
        For the expectation of $Q_n$ it holds:
        \begin{eqnarray}        
                E[Q_n]&=& E[\sum_{i=1}^{n} \sum_{k=1}^{\sum_{l=0}^{Y_i}S_{i,l}}Z_{n,i,k}]\\
                &=& \sum_{i=1}^{n}E[ \sum_{k=1}^{\sum_{l=0}^{Y_i}S_{i,l}}Z_{n,i,k}]\\
                &=& \sum_{i=1}^{n}      E[ E[ \sum_{k=1}^{\sum_{l=0}^{Y}S_{i,l}}Z_{n,i,k}|\lambda, \sigma, \mu]]\\                 
        E[ \sum_{k=1}^{\sum_{l=0}^{Y}S_{i,l}}Z_{n,i,k}|\lambda, \sigma, \mu]
                        &=& E[ E[ \sum_{k=1}^{\sum_{l=0}^{Y}S_{i,l}}Z_{n,i,k}|\sum_{l=0}^{Y_i}S_{i,l}]|\lambda, \sigma, \mu]\\ 
                &=& E[ \sum_{l=0}^{Y_i}S_{i,l}|\lambda, \sigma, \mu]  E[ Z_{n,i,1}|\lambda, \sigma, \mu]\\
                &=&E[Y_1|\lambda,\mu]E[S_{1,1}| \sigma]E[Z_{n,i,1}|\mu] 
        \end{eqnarray}

 \item          $B_n$:
 By definition, it holds                
 \begin{eqnarray}                        
 E[B_n] &=& E[\sum_{j=1}^C Z_{n,0,k}] = C E[Z_{n,0,k}]
 \end{eqnarray}
 
        \item $D_i$:
                If all $Z_{i,j,k}$, $Y_1$ and $S_{i,l}$ are uncorrelated, it holds
                \begin{eqnarray}                         
                E[D_i]&=& E[D_{i-1} ( 1- \Pi_{j=0}^{i-1} \Pi_{k=0}^{\sum_{l=0}^{Y_j}S_{j,l}} (1-Z_{i,j,k}))]\\
                &=& E[D_{i-1}] (1 -  E[ \Pi_{j=0}^{i-1} \Pi_{k=0}^{\sum_{l=0}^{Y_j}S_{j,l}} (1-Z_{i,j,k})])\\
                &=&E[D_{i-1}] ( 1 -  E[E[ \Pi_{j=0}^{i-1} \Pi_{k=0}^{\sum_{l=0}^{Y_j}S_{j,l}} (1-Z_{i,j,k})|  Y, S]])\\
                &=&E[D_{i-1}] ( 1 -  E[\Pi_{j=0}^{i-1} \Pi_{k=0}^{\sum_{l=0}^{Y_j}S_{j,l}} (1-E[Z_{i,j,k}])])\\    
                &=&E[D_{i-1}] ( 1 -  E[\Pi_{j=0}^{i-1} (1-E[Z_{i,j,k}])^{\sum_{l=0}^{Y_j}S_{j,l}}])\\
                &\leq& E[D_{i-1}] (  1 -  \Pi_{j=0}^{i-1} (1-E[Z_{i,j,k}])^{E[\sum_{l=0}^{Y_j}S_{j,l}]})\\
                &=&E[D_{i-1}] (  1 -  \Pi_{j=0}^{i-1} (1-E[Z_{i,j,k}])^{ E[Y_1]E[S_{1,1}]})\\
                E[D_1] &=& (1-(1-E[Z_{1,0,1}])^{C})
                \end{eqnarray}
                where the third line follows by the law of total probability and the $6$'th by Jensens Inequality.

\end{itemize}   
        \end{proof}

        \begin{proof}[Proof of Proposition \ref{VARFORM}]
                
         With $M_n$,$D_n$,$Q_n$ and $B_n$ uncorrelated, it holds for the variance of $L_n$:
                \begin{eqnarray}
                V[L_n] &=&V[M_n D_n (Q_n+B_n)]\\
                &=& E[M_n]^2 V[D_n (Q_n+B_n)] + V[M_n]E[D_n (Q_n+B_n)]^2\\
                &&+V[M_n]V[D_n (Q_n+B_n)]\\
                \end{eqnarray}
                
                         and further:
                         \begin{eqnarray}               
                         V[D_n (Q_n+B_n)] &=& E[D_n]^2 (V[Q_n]+V[B_n]) + (E[Q_n]+E[B_n])^2 V[D_n]\\
                         &&+ V[D_n] (V[Q_n]+V[B_n])\\
                         \end{eqnarray}
                         
                         For the variance of $Q_n$ it holds:
                         
                         \begin{eqnarray}
                         V[Q_n]&=& E[V[Q_n|\lambda,\sigma,\mu]] +V[E[Q_n|\lambda,\sigma,\mu]] 
                         \end{eqnarray} 
                         
                         and

                         \begin{align}
                         V[Q_n|\lambda,\sigma,\mu] =& \sum_{i=1}^{n} \left( V[\sum_{l=0}^{Y_i}S_{i,l}|\lambda,\mu,\sigma]E[Z_{n,i,1}|\mu]^2+E[\sum_{l=0}^{Y_i}S_{i,l} |\lambda,\mu,\sigma] V[Z_{n,i,1}|\mu]\right)\\
                         =&\sum_{i=1}^{n} \left((V[Y_i] E[S_{i,l}]^2 + E[Y_i] V[S_{i,l}]) E[Z_{n,i,1}]^2+E[Y_i]E[S_{i,l}] V[Z_{n,i,1}]\right)
                         \end{align}
                         by the law of total variance.

                         For $B_n$ we can now use the law of total variance to obtain:  
                         \begin{eqnarray}
                         V[B_n] &=&  V[ \sum_{j= 1} ^C Z_{n,i,k}] \\
                         &=& E[ V[ \sum_{j= 1} ^C Z_{n,i,k}|\mu]] +  V[E[ \sum_{j= 1} ^C Z_{n,i,k}|\mu]]\\
                         &=& E[ C V[ Z_{n,i,k}|\mu]] +  V[ C E[Z_{n,i,k}|\mu]]\\
                         &=& CE[  V[ Z_{n,i,k}|\mu]] +C^2 V[E[Z_{n,i,k}|\mu]] 
                         \end{eqnarray}

  Lastly, for $D_n$, note that $E[D_n^2] = E[D_n]$ because $D_n \in \left\lbrace 0,1 \right\rbrace$. It follows 
  \begin{eqnarray}      
   V[D_n] &=& E[D_n]-E[D_n]^2
  \end{eqnarray}


        \end{proof}

        \section{Data Trace} \label{fitting}
         To have a better understanding of the scaling behavior of real deployments and to create a model suitable for simulating clusters, we fitted 
        the behavior of deployments to a real-world data trace. The particular data trace we use was published by Cortez et al.~\citeyear{Cortez}.  This dataset consists of one month of data of internal Microsoft Azure jobs. It contains 35,576 deployments,\footnote{In contrast to \cite{Cortez} we did not consolidate all deployments a single user runs on a certain day into one. This is because cores that get requested as a new deployment do not need to be accepted on the same cluster.} though only 29,757 of these deployments arrived during the observed time period. Since we want to fit distributions with the goal of simulating arriving deployments, only these 29,757 deployments can be used for most of our fitting. The deployments that arrived before the beginning of the observed period of time cannot be used when making maximum likelihood estimations, because for start times before the observed period of time, only longer lived deployments survived to be observed. Including them would strongly skew the fit.
        The 29,757 deployments activated 4,317,961 cores, out of which 4,211,926 became inactive again during the observed month.
        The exact lifetime of the remaining cores (i.e., the length of time between becoming active and then inactive again) is not known; instead we only have a lower bound on it (i.e., our observation is Type I censored: see for example \cite{EngSta}). Thus, for cores where we only have a lower bound on the lifetime we use the cdf in our likelihood function while for cores whose lifetime is known we use the pdf.

        \subsection{Fitting on the Deployment Level}
        \label{sec:deploymentfit}
        We first fit arrival and departure processes for each individual deployment.
        In keeping with the Markov assumption, we fit a Poisson distribution to the scale out rate of each deployment, while we fit an exponential distribution to the lifetime of cores for each deployment for which at least one core became inactive during the observed time period.
        Note that while we model the cluster admission problem as a discrete-time POMDP, the processes are fit in continuous time. This is more general and avoids imprecisions introduced by time discretization.
        To fit the size of a scale out, we also used a Poisson distribution (plus 1, as scale outs must have at least one core).\footnote{As the Poisson distribution is single-parameter and its variance cannot be set independent of the average size, this is not a particularly good fit for users with large but consistent scale out sizes. However, its simplicity avoids overfitting on the often low number of samples per deployment and it results in a good fit on the population level.}
We further assume that each deployment, had it lived forever, at some point would have made a scale out request for more than 1 core. Since we did not observe these scale-outs and therefore cannot make a direct likelihood fit, we introduce two parameters $P_1$ and $P_2$ to represent them. We assume that the scale out rate of deployments that never scaled out (some because they died, but many simply because the observation period of the dataset ended) is equal to the value for which not observing a scale out has probability $P_1$. We equivalently set the scale out size for deployments that never increase their size by more than 1 core during one scale out event according to $P_2$. We calibrated $P_1$ and $P_2$ by minimizing the (discrete) Cram\'er-von Mises distance of the size of deployments between samples drawn from our fitted model and the data set. The optimal distance is $0.1585$ and an overlay of both cumulative distribution functions can be seen in Figure \ref{cramer1}. Note that most of the remaining distance does not seem to be caused by limitations of our model or fitting procedure, but by limitations of the dataset. The dataset, while relatively large, still does contain a somewhat small selection of deployments from the tail. More importantly, it only contains \emph{internal} Azure deployments, so the types of workloads are limited. As such, it contains few deployments of sizes between $100$ and $1500$, but a relatively large number of deployments of sizes between $1500$ and $2000$. This effect is visualized in Figure \ref{cramer2}, which shows the CDF over the percentage of utilization in the cluster coming from deployments of different sizes for both our model and the dataset. 
        
                \begin{figure}[t]
                        \centering%
                        \begin{minipage}{0.5\textwidth}
                                \centering
                                \includegraphics[height=4.5cm]{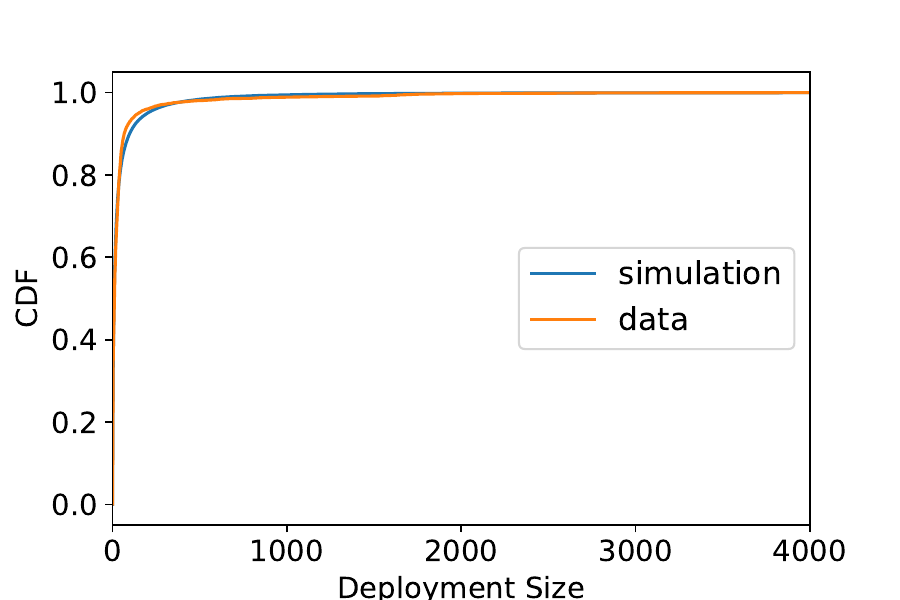}
                                \caption{CDF over number of\\ deployments of all sizes}
                                \label{cramer1}
                        \end{minipage}%
                        \begin{minipage}{0.5\textwidth}
                                \centering
                                \includegraphics[height=4.5cm]{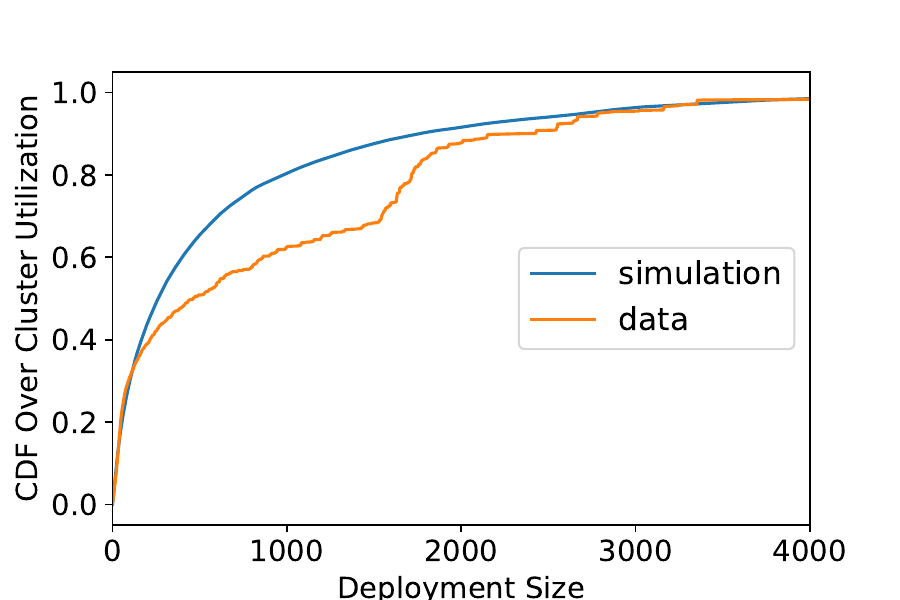}
                                \caption{CDF over utilization percentage\\ from deployments of all sizes}
                                \label{cramer2}
                        \end{minipage}%
                \end{figure}%
        
                \begin{table*}[b]
                        \small
                        
                        \centering
                        
                        \begin{tabular}{cc}
                                \hline
                                \multirow{3}{*}{Priors} & $\mu \approx Gamma(0.3107,0.5778)$ \\
                                
                                & $\lambda \approx Gamma(0.4907,0.4496)$\\
                                
                                &$\sigma \approx Gamma(0.2616,0.0552)$\\
                                \hline \noalign{\smallskip                              }
                                
                                \hline 
                                \multirow{2}{*}{Global Parameters} & $\Delta = 0.119$ \\
                                
                                &  $\nu = 0.673$   \\
                                \hline \noalign{\smallskip}                     
                        \end{tabular}
                        \caption{Fitted processes}      \label{table1a}
                \end{table*}

        \subsection{Fitting on the Population Level}
        \label{sec:FittingOnThePopulationLevel}

      With the distributions for each deployment in place, we now fit Gamma distributions for the population. The parameters of the processes for each arriving deployment are drawn from these populations. 
        As the data was skewed, positive, and not really heavy tailed, a Gamma distribution is a natural and very general candidate (containing the Chi-squared, Erlang and Exponential distributions as special cases), with the added benefit of being conjugate prior to the deployment processes. The resulting model and parameters from our fits are shown in Table~\ref{table1a}. While the scale out size is fit directly to the samples, scale out rate and core lifetime are highly correlated. The longer a deployment's cores live, the lower the rate at which new cores arrive, as can be seen in Figure \ref{lifespawn}. This shows that deployments with long lived cores do not necessarily have more active cores. To account for this, we fit the power law relationship $\nu$ between scale out rate and lifetime, i.e.,  we fitted the prior distribution on scale out rates multiplied by the respective core lifetimes taken to the power of $\nu$. We have chosen $\nu$ such that the mean absolute distance between normalized scale out rate of each deployment and the average (normalized) scale out rate is minimized. 
        
        To visualize the fitted distributions, Figure \ref{lifegamma} shows the
        CDF of the Gamma distribution for the lifetime parameter, overlaid over the normalized cumulative histogram of the fitted rates of the sample deployments. 
        
        \begin{figure}[t]
        \centering%
                        \begin{minipage}{0.5\textwidth}
                                \centering
                                \includegraphics[height=4.5cm]{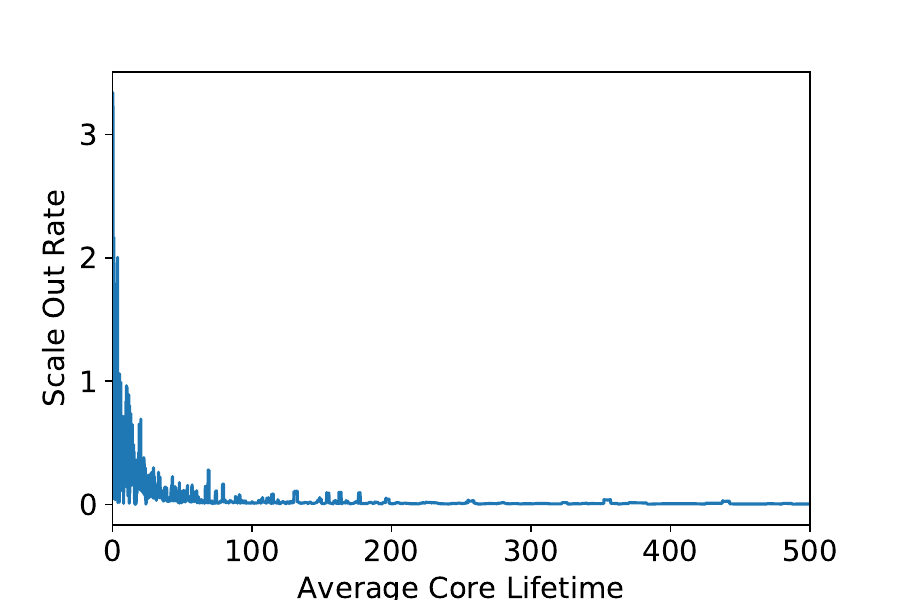}
                                \caption{Scale out rate as a  function\\ of average core lifetime}
                                \label{lifespawn}
                        \end{minipage}%
                \begin{minipage}{0.5\textwidth}
                        \centering
                        \includegraphics[height=4.5cm]{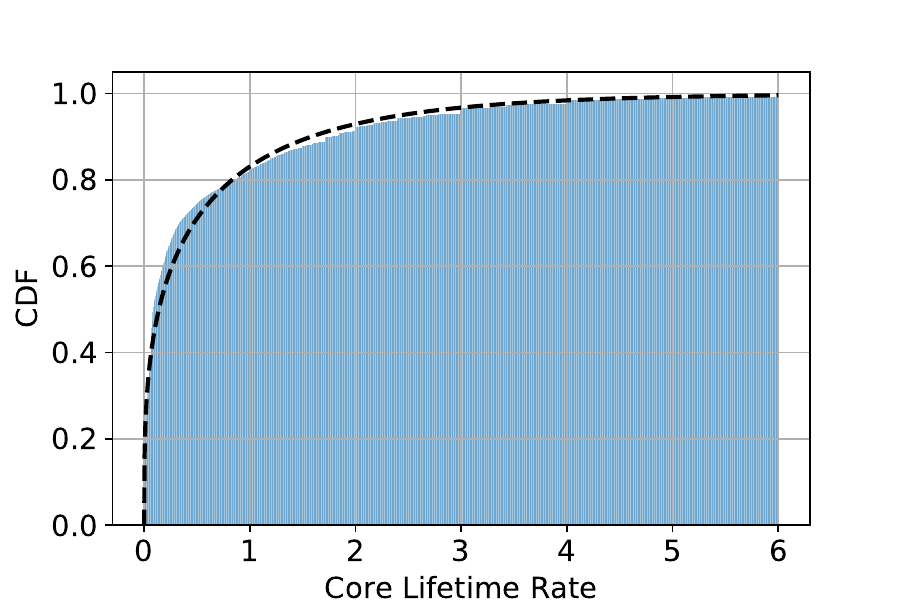}
                        \caption{Distribution of the core lifetime\\ rate parameter}
                        \label{lifegamma}
                \end{minipage}%
        \end{figure}%
        
        Figure \ref{rategamma} shows the CDF for the normalized scale out rates over the relevant cumulative histogram. The actual scale out rate of a sampled deployment is now simply the normalized scale out rate multiplied by the average core lifetime. Figure \ref{sizegamma} shows the fitted CDF for the scale out size parameter over its cumulative histogram. 

        \begin{figure}[t]
                \centering%
                \begin{minipage}{0.48\textwidth}
                        \centering
                        \includegraphics[height=4.5cm]{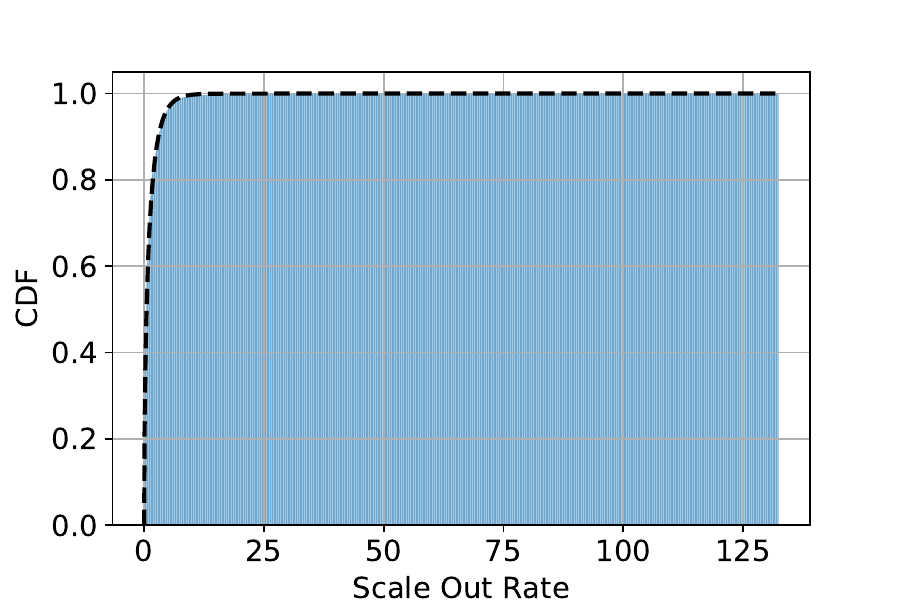}
                        \caption{Distribution of the scale  out\\ rate parameter}
                        \label{rategamma}
                        
                \end{minipage}%
                                \begin{minipage}{0.48\textwidth}
                                        \centering                                      
                                        \includegraphics[height=4.5cm]{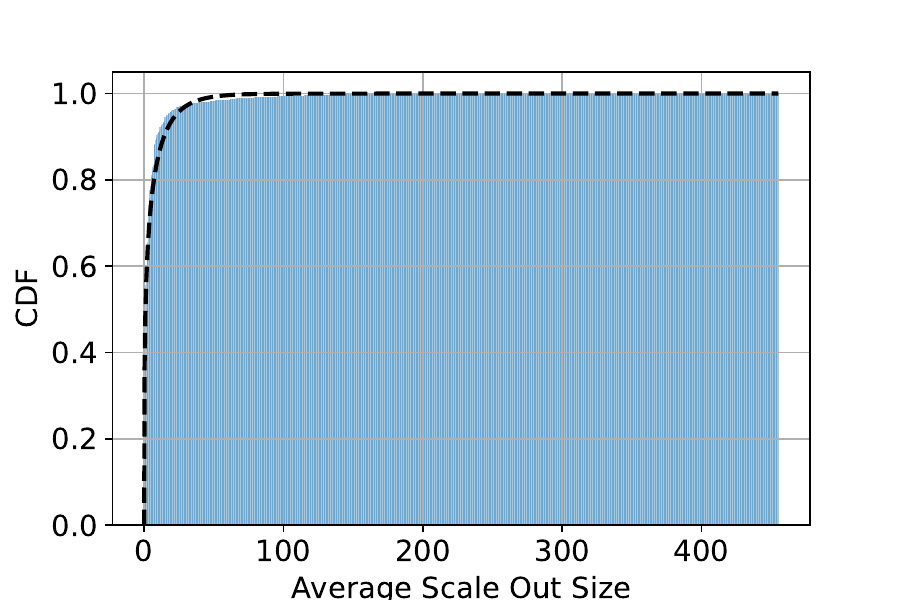}
                                        \caption{Distribution of the scale out\\ size parameter}
                                        \label{sizegamma}
                                \end{minipage}%
        \end{figure}%

        \paragraph{Deployment Shutdown.}
        While most deployments in the dataset die because they have zero active cores, 5,980 of the 22,241 deployments that both arrive and die during the observed period seem to get actively shut down.  By this we mean that they had at least 3 VMs that all shut down simultaneously. This would be highly unlikely if deployments only die when cores or VMs become inactive independently. To capture such behavior we fit an exponential distribution over the number of expected core lifetime deployments lived. The maximal lifetime of deployments that did not get shut down was assumed to be censored to their realized lifetime.
        
        \section{Moment Approximation with Gamma Priors}\label{A:GAMMA}
                                        \begin{proposition}\label{PROP:GAMMAPRIOR}

                                                When $Y_i\sim Pois(\lambda \mu^\nu)$, $\lambda \sim Gamma(a,b)$, $S_{i,l}\sim Pois(\sigma)$, $\sigma \sim Gamma(\alpha, \beta)$, $Z_{n,i,j} \sim Bernoulli(e^{(i-n)\mu})$ (Bernoulli over complementary CDF of an exponential distribution), $\mu \sim Gamma(\mathfrak{a},\mathfrak{b})$ and $M_i \sim Bernoulli(e^{(i-n)\Delta \mu})$  it holds:
                                                
                                                \begin{eqnarray}
                                                E[Q_n]&=&\frac{a}{b} \frac{\alpha +\beta}{\beta} \frac{ \Gamma(\mathfrak{a}+\nu)}{\Gamma(\mathfrak{a})}   \sum_{i=1}^{n}\frac{  \mathfrak{b}^\mathfrak{a}  }{(n+ \mathfrak{b}-i)^{\mathfrak{a}+\nu}}\\
                                                E[D_i] &\leq& E[D_{i-1}]( 1 -  \Pi_{j=0}^{i-1} (1-(1+\frac{i-j}{\mathfrak{b}})^{-(\mathfrak{a})})^{\frac{a}{b} \frac{\alpha }{\beta} })\\
                                                E[Z_{n,i,1}] &=& \frac{\mathfrak{b}^\mathfrak{a}}{(n+ \mathfrak{b}-i)^{\mathfrak{a}} }         \\
                                                E[M_n] &=& \frac{\mathfrak{b}^\mathfrak{a}}{(\Delta n+ \mathfrak{b})^{\mathfrak{a}} }         
                                                \end{eqnarray}
                                                and
                                                \begin{eqnarray}
                                                V[Q_n]&=&   \mathfrak{b}^\mathfrak{a} \frac{ \Gamma(\mathfrak{a}+\nu)}{\Gamma(\mathfrak{a})} \left[ \left( \frac{a}{b}  (\frac{\alpha}{\beta^2} +\frac{\alpha+\beta }{\beta}^2 -1)\ \right) \right. \\
                                                && \left. \sum_{i=1}^{n}\frac{1}{(2n+ \mathfrak{b}-2i)^{\mathfrak{a}+\nu} } +  \frac{a}{b} \frac{\alpha+\beta }{\beta} \sum_{i=1}^{n} \frac{1}{(n+ \mathfrak{b}-i)^{\mathfrak{a}+\nu}} \right] \\
                                                &&+ \left(\frac{a}{b}^2 \frac{\alpha + \beta}{\beta}^2  + (\frac{a}{b}^2  \frac{\alpha}{\beta^2} + \frac{a}{b^2 }\frac{\alpha+\beta}{\beta}^2+ \frac{a}{b^2 }\frac{\alpha}{\beta^2}) \right) \\
                                                &&\left[ \mathfrak{b}^\mathfrak{a}  \frac{ \Gamma(\mathfrak{a}+2\nu)}{\Gamma(\mathfrak{a})} \sum_{1\leq i \leq j \leq n}  \frac{1}{(2n+ \mathfrak{b}-i-j)^{\mathfrak{a}+2\nu} } \right.\\
                                                &&\left.  - (\mathfrak{b}^\mathfrak{a} \frac{ \Gamma(\mathfrak{a}+\nu)}{\Gamma(\mathfrak{a})})^2  \sum_{1\leq i \leq j \leq n} \frac{1}{(n+ \mathfrak{b}-i)^{\mathfrak{a}+\nu} } \frac{1}{(n+ \mathfrak{b}-j)^{\mathfrak{a}+\nu} } \right]\\
                                                &&  + (\frac{a}{b}^2  \frac{\alpha}{\beta^2} + \frac{a}{b^2 }\frac{\alpha+\beta}{\beta}^2+ \frac{a}{b^2 }\frac{\alpha}{\beta^2})  \left[\frac{ \Gamma(\mathfrak{a}+\nu)}{\Gamma(\mathfrak{a})}\sum_{i=1}^{n} \frac{\mathfrak{b}^\mathfrak{a}}{(n+ \mathfrak{b}-i)^{\mathfrak{a}+\nu} }\right]^2
                                                \end{eqnarray}
                                                \begin{eqnarray}
                                                V[Z_{n,i,1}] &=& \frac{\mathfrak{b}^\mathfrak{a}}{(n+ \mathfrak{b}-i)^{\mathfrak{a}} } (1 - \frac{\mathfrak{b}^\mathfrak{a}}{(n+ \mathfrak{b}-i)^{\mathfrak{a}} })\\                                             
                                                CE[  V[ Z_{n,i,k}|\mu]] +C^2 V[E[Z_{n,i,k}|\mu]]                                                                                             
                                                &=& C \frac{\mathfrak{b}^\mathfrak{a}}{(n+ \mathfrak{b}-i)^{\mathfrak{a}} } (1 - C \frac{\mathfrak{b}^\mathfrak{a}}{(n+ \mathfrak{b}-i)^{\mathfrak{a}} }) \\&&+ (C^2-C) \frac{\mathfrak{b}^\mathfrak{a}}{(2n+ \mathfrak{b}-2i)^{\mathfrak{a}} }               \\                                                                                                                                                                              
                                                V[M_n] &=& \frac{\mathfrak{b}^\mathfrak{a}}{(\Delta n+ \mathfrak{b})^{\mathfrak{a}} } (1-\frac{\mathfrak{b}^\mathfrak{a}}{(\Delta n+ \mathfrak{b})^{\mathfrak{a}} })
                                                \end{eqnarray}                          %
                                                
                                                \end{proposition}

                                                \begin{proof}
                                                        \begin{itemize}
   \allowdisplaybreaks                                                                 \item For $Q_n$ it holds
                                                                \begin{eqnarray}
                                                                E[Q_n]&=& \sum_{i=1}^{n} E[E[Y_1|\lambda, \mu]E[S_{1,1}| \sigma]E[Z_{n,i,1}|\mu]]\\
                                                                &=& \sum_{i=1}^{n} E[Y_1|\lambda, \mu]E[S_{1,1}| \sigma]E[Z_{n,i,1}|\mu]\\
                                                                &=& \sum_{i=1}^{n} \lambda \mu^\nu  (\sigma +1)  e^{(i-n)\mu}
                                                                \end{eqnarray}
                                                                \begin{eqnarray}
                                                                E[\lambda]&=& \frac{a}{b}\\   
                                                                V[\lambda]&=&         \frac{a}{b^2}\\                 
                                                                E[\sigma+1]&=& \frac{\alpha+\beta }{\beta} \\
                                                                V[\sigma+1]&=& \frac{\alpha}{\beta^2} \\
                                                                E[\mu^\nu e^{(i-n)\mu}] &=&  \int_0^\infty \mu^\nu e^{(i-n)\mu} \frac{\mathfrak{b}^\mathfrak{a}\mu^{\mathfrak{a}-1} e^{- \mathfrak{b} \mu} }{\Gamma(\mathfrak{a})}d\mu\\
                                                                &=&  \frac{\mathfrak{b}^\mathfrak{a} }{\Gamma(\mathfrak{a})}\int_0^\infty  \mu^{\mathfrak{a}-1+\nu} e^{(i-n- \mathfrak{b}) \mu} d\mu\\
                                                                &=&  \frac{\mathfrak{b}^\mathfrak{a} }{\Gamma(\mathfrak{a})} (n+ \mathfrak{b}-i)^{-\mathfrak{a}-\nu} \Gamma(\mathfrak{a}+\nu)\\
                                                                &=& \frac{\mathfrak{b}^\mathfrak{a}}{(n+ \mathfrak{b}-i)^{\mathfrak{a}+\nu} }\frac{ \Gamma(\mathfrak{a}+\nu)}{\Gamma(\mathfrak{a})}
                                                                \end{eqnarray}
                                                                %
                                                                It immediately follows:
                                                                \begin{eqnarray}
                                                                E[Q_n]&=&\frac{a}{b} \frac{\alpha +\beta}{\beta}    \frac{ \Gamma(\mathfrak{a}+\nu)}{\Gamma(\mathfrak{a})} \sum_{i=1}^{n}\frac{\mathfrak{b}^\mathfrak{a} }{(n+ \mathfrak{b}-i)^{\mathfrak{a}+\nu}}                                         
                                                                \end{eqnarray}
                                                                Next we will calculate 
                                                                
                                                                \begin{eqnarray}
                                                                V[Q_n]&=& E[V[Q_n|\lambda,\sigma,\mu]] \\
                                                                &&+V[E[Q_n|\lambda,\sigma,\mu]]\\       
                                                                \end{eqnarray}
                                                                Before we can do so, we need to collect a few easy supporting results:
                                                                \begin{eqnarray}
                                                                V[Y_1|\lambda,\mu]      &=& \lambda \mu^\nu  \\ 
                                                                V[S_{1,1}|\sigma]&=& \sigma\\
                                                                E[Z_{n,i,1}|\mu]^2 &=& e^{((i-n)\mu)^2} \\
                                                                &=& e^{(2i-2n)\mu} \\
                                                                &=& E[Z_{2n,2i,1}|\mu]\\        
                                                                V[Z_{n,i,1}|\mu] &=& e^{((i-n)\mu)}(1-e^{((i-n)\mu)} )\\
                                                                &=&  e^{((i-n)\mu)}-e^{2((i-n)\mu)} \\
                                                                &=& E[Z_{n,i,1}|\mu]-E[Z_{2n,2i,1}|\mu]\\
                                                                \end{eqnarray}
                                                                We also need
                                                                \begin{eqnarray}
                                                                E[\lambda^2]&=& V[\lambda]+E[\lambda]^2\\
                                                                &=& \frac{a}{b^2}+\frac{a}{b}^2\\
                                                                E[(\sigma+1)^2]&=&V[\sigma+1]+E[\sigma+1]^2\\
                                                                &=& \frac{\alpha}{\beta^2} +\frac{\alpha+\beta }{\beta}^2 \\
                                                                E[\mu^{2\nu} e^{(2i-2n)\mu}] &=& \int_0^\infty \mu^{2\nu} e^{(2i-2n)\mu} \frac{\mathfrak{b}^\mathfrak{a}\mu^{\mathfrak{a}-1} e^{- \mathfrak{b} \mu} }{\Gamma(\mathfrak{a})}d\mu\\
                                                                &=&  \frac{\mathfrak{b}^\mathfrak{a} }{\Gamma(\mathfrak{a})}\int_0^\infty  \mu^{\mathfrak{a}-1+2\nu} e^{(2i-2n- \mathfrak{b}) \mu} d\mu\\
                                                                &=&  \frac{\mathfrak{b}^\mathfrak{a} }{\Gamma(\mathfrak{a})} (2n+ \mathfrak{b}-2i)^{-\mathfrak{a}-2\nu} \Gamma(\mathfrak{a}+2\nu)\\
                                                                &=& \frac{ \Gamma(\mathfrak{a}+2\nu)}{\Gamma(\mathfrak{a})} \frac{\mathfrak{b}^\mathfrak{a}}{(2n+ \mathfrak{b}-2i)^{\mathfrak{a}+2\nu} }                                  
                                                                \end{eqnarray}  
                                                                This now allows us to calculate everything that is needed for the first half of the variance of $Q_n$, i.e.,  $E[V[Q_n|\lambda,\sigma,\mu]]$. First note that
                                                                \begin{eqnarray}
                                                                V[Q_n|\lambda,\sigma,\mu] &=&\sum_{i=1}^{n} ( \left( V[Y_i|\lambda,\mu] E[S_{i,l}|\sigma]^2 \right. \\
                                                                &&\left. + E[Y_i|\lambda,\mu] V[S_{i,l}|\mu]\right) E[Z_{n,i,1}|\mu]^2\\
                                                                &&+E[Y_i|\lambda]E[S_{i,l}|\sigma] V[Z_{n,i,1}|\mu])
                                                                \end{eqnarray}
                                                                and 
                                                                \begin{align}
                                                                E[(V[Y_i|\lambda,\mu] E[S_{i,l}|\sigma]^2  E[Z_{n,i,1}|\mu]^2 ] =& E[(\lambda)(\sigma+1)^2 \mu^\nu e^{(2i-2n)\mu}]\\
                                                                =& E[\lambda]E[(\sigma+1)^2]E[\mu^\nu e^{(2i-2n)\mu}]\\       
                                                                =& \frac{a}{b}  (\frac{\alpha}{\beta^2} +\frac{\alpha+\beta }{\beta}^2 ) \frac{ \Gamma(\mathfrak{a}+\nu)}{\Gamma(\mathfrak{a})}  \frac{\mathfrak{b}^\mathfrak{a}}{(2n+ \mathfrak{b}-2i)^{\mathfrak{a}+\nu} }                                                 
                                                                \end{align}
                                                                
                                                                \begin{eqnarray}
                                                                E[E[Y_i|\lambda,\mu] V[S_{i,l}|\mu]) E[Z_{n,i,1}|\mu]^2] &=& E[\lambda \mu^\nu \sigma e^{(2i-2n)\mu} ]\\
                                                                &=& E[\lambda]E[\sigma]E[\mu^\nu e^{(2i-2n)\mu}]\\
                                                                &=& \frac{a}{b} \frac{\alpha}{\beta} \frac{ \Gamma(\mathfrak{a}+\nu)}{\Gamma(\mathfrak{a})} \frac{\mathfrak{b}^\mathfrak{a}}{(2n+ \mathfrak{b}-2i)^{\mathfrak{a}+\nu} }       
                                                                \end{eqnarray}
                                                                
                                                                and                                                             
                                                            \begin{align}
                                                                E[E[Y_i|\lambda,\mu]E[S_{i,l}|\sigma] V[Z_{n,i,1}|\mu]] =& E[\lambda \mu^\nu (\sigma+1) (E[Z_{n,i,1}|\mu]-E[Z_{2n,2i,1}|\mu]) ]\\
                                                                =& E[\lambda \mu^\nu (\sigma+1) (e^{(i-n)\mu}-e^{(2i-2n)\mu}) ]\\
                                                                =& E[\lambda]E[\sigma+1](\mu^\nu E[e^{(i-n)\mu}]-E[\mu^\nu e^{(2i-2n)\mu}])      \\
                                                                =& \frac{a}{b} \frac{\alpha+\beta }{\beta} \frac{ \Gamma(\mathfrak{a}+\nu)}{\Gamma(\mathfrak{a})} \left(\frac{ \mathfrak{b}^\mathfrak{a}}{(n+ \mathfrak{b}-i)^{\mathfrak{a}+\nu}} - \frac{\mathfrak{b}^\mathfrak{a}}{(2n+ \mathfrak{b}-2i)^{\mathfrak{a}+\nu}} \right)
                                                                \end{align}  
                                                                
                                                                It follows: 
                                                                \begin{align}
                                                                E[(V[Y_i|\lambda,\mu] E[S_{i,l}|\sigma]^2  E[Z_{n,i,1}|\mu]^2 ] =& \mathfrak{b}^\mathfrak{a} \frac{ \Gamma(\mathfrak{a}+\nu)}{\Gamma(\mathfrak{a})} \left[ \left( \frac{a}{b}  (\frac{\alpha}{\beta^2} +\frac{\alpha+\beta }{\beta}^2 ) +\frac{a}{b} \frac{\alpha}{\beta} -\frac{a}{b} \frac{\alpha+\beta }{\beta} \right) \right. \\
   & \left. \frac{1}{(2n+ \mathfrak{b}-2i)^{\mathfrak{a}+\nu} } +  \frac{a}{b} \frac{\alpha+\beta }{\beta}  \frac{1}{(n+ \mathfrak{b}-i)^{\mathfrak{a}+\nu}} \right] \\
 =& \mathfrak{b}^\mathfrak{a}  \frac{ \Gamma(\mathfrak{a}+\nu)}{\Gamma(\mathfrak{a})} \left[ \left( \frac{a}{b}  (\frac{\alpha}{\beta^2} +\frac{\alpha+\beta }{\beta}^2 -1)\ \right) \right. \\
& \left. \frac{1}{(2n+ \mathfrak{b}-2i)^{\mathfrak{a}+\nu} } +  \frac{a}{b} \frac{\alpha+\beta }{\beta}  \frac{1}{(n+ \mathfrak{b}-i)^{\mathfrak{a}+\nu}} \right] 
                                                                \end{align}
                                                                
                                                                Finally for the second part of the variance, i.e., $V[E[Q_n|\lambda,\sigma,\mu]]$, we need:
                                                                
                                                                \begin{align}
                                                                V[\sum_{i=1}^{n} \mu^\nu e^{(i-n)\mu}] =& \sum_{1\leq i \leq j \leq n}  Cov[ \mu^\nu e^{(i-n)\mu},  \mu^\nu e^{(j-n)\mu}]\\
                                                                =&  \sum_{1\leq i \leq j \leq n}  E[ \mu^\nu e^{(i-n)\mu} \mu^\nu e^{(j-n)\mu}]-E[\mu^\nu e^{(i-n)\mu}]E[\mu^\nu e^{(j-n)\mu}]\\
                                                                =&  \sum_{1\leq i \leq j \leq n}  E[\mu^{2\nu} e^{(i+j-2n)\mu}]-E[ \mu^\nu e^{(i-n)\mu}]E[\mu^\nu e^{(j-n)\mu}]\\
                                                                =&  \sum_{1\leq i \leq j \leq n} ( \frac{\mathfrak{b}^\mathfrak{a}}{(2n+ \mathfrak{b}-i-j)^{\mathfrak{a}+2\nu} } \frac{ \Gamma(\mathfrak{a}+2\nu)}{\Gamma(\mathfrak{a})}\\
                                                                & - \frac{ (\Gamma(\mathfrak{a}+\nu)}{\Gamma(\mathfrak{a})})^2\frac{\mathfrak{b}^\mathfrak{a}}{(n+ \mathfrak{b}-i)^{\mathfrak{a}+\nu} } \frac{\mathfrak{b}^\mathfrak{a}}{(n+ \mathfrak{b}-j)^{\mathfrak{a}+\nu} })\\
                                                                =& \mathfrak{b}^\mathfrak{a}  \frac{ \Gamma(\mathfrak{a}+2\nu)}{\Gamma(\mathfrak{a})} \sum_{1\leq i \leq j \leq n}  \frac{1}{(2n+ \mathfrak{b}-i-j)^{\mathfrak{a}+2\nu} }\\
                                                                & - (\mathfrak{b}^\mathfrak{a} \frac{ \Gamma(\mathfrak{a}+\nu)}{\Gamma(\mathfrak{a})})^2  \sum_{1\leq i \leq j \leq n} \frac{1}{(n+ \mathfrak{b}-i)^{\mathfrak{a}+\nu} } \frac{1}{(n+ \mathfrak{b}-j)^{\mathfrak{a}+\nu} }
                                                                \end{align}
                                                                \begin{eqnarray}
                                                                V[\lambda  (\sigma +1)]&=& E[\lambda]^2 V[\sigma+1]+V[\lambda] E[\sigma+1]^2+V[\lambda] V[\sigma+1]\\
                                                                &=&\frac{a}{b}^2  \frac{\alpha}{\beta^2} + \frac{a}{b^2 }\frac{\alpha+\beta}{\beta}^2+ \frac{a}{b^2 }\frac{\alpha}{\beta^2}
                                                                \end{eqnarray}          
                                                                Now we can write:                                  
                                                                \begin{align}
   V[E[Q_n|\lambda,\sigma,\mu]] =& V[\sum_{i=1}^{n} E[ \sum_{k=1}^{\sum_{l=0}^{Y}S_{i,l}}Z_{n,i,k}|\lambda, \sigma, \mu]]\\
                                                                =& V[\sum_{i=1}^{n} \lambda \mu^\nu (\sigma +1)  e^{(i-n)\mu}]\\
                                                                =& V[ \lambda (\sigma +1)  \sum_{i=1}^{n} \mu^\nu e^{(i-n)\mu}]\\
                                                                =& E[\lambda  (\sigma +1)]^2 V[\sum_{i=1}^{n} \mu^\nu e^{(i-n)\mu}] +V[\lambda  (\sigma +1)] V[\sum_{i=1}^{n}  \mu^\nu e^{(i-n)\mu}] \\
                                                                &  + V[\lambda  (\sigma +1)] E[\sum_{i=1}^{n}  \mu^\nu e^{(i-n)\mu}]^2\\
                                                                =& \left(\frac{a}{b}^2 \frac{\alpha + \beta}{\beta}^2  + (\frac{a}{b}^2  \frac{\alpha}{\beta^2} + \frac{a}{b^2 }\frac{\alpha+\beta}{\beta}^2+ \frac{a}{b^2 }\frac{\alpha}{\beta^2}) \right) \\
                                                                &\left( \mathfrak{b}^\mathfrak{a}  \frac{ \Gamma(\mathfrak{a}+2\nu)}{\Gamma(\mathfrak{a})} \sum_{1\leq i \leq j \leq n}  \frac{1}{(2n+ \mathfrak{b}-i-j)^{\mathfrak{a}+2\nu} } \right.\\
                                                                &\left.  - (\mathfrak{b}^\mathfrak{a} \frac{ \Gamma(\mathfrak{a}+\nu)}{\Gamma(\mathfrak{a})} )^2  \sum_{1\leq i \leq j \leq n} \frac{1}{(n+ \mathfrak{b}-i)^{\mathfrak{a}+\nu} } \frac{1}{(n+ \mathfrak{b}-j)^{\mathfrak{a}+\nu} } \right)\\
                                                                &  + (\frac{a}{b}^2  \frac{\alpha}{\beta^2} + \frac{a}{b^2 }\frac{\alpha+\beta}{\beta}^2+ \frac{a}{b^2 }\frac{\alpha}{\beta^2})  \left(\frac{ \Gamma(\mathfrak{a}+\nu)}{\Gamma(\mathfrak{a})} \sum_{i=1}^{n} \frac{\mathfrak{b}^\mathfrak{a}}{(n+ \mathfrak{b}-i)^{\mathfrak{a}+\nu} }\right)^2                              
                                                                \end{align}
                                                                Inserting into Propositions 1 and 2 now yields the result.                
                                                                
                                                                \item For $D_i$ note the following:                                       
                                                                As an exponential distribution whose rate is drawn from a Gamma distribution with shape $\mathfrak{a}$ and rate $\mathfrak{b}$ is equal to a Lomax distribution with scale $\mathfrak{b}$ and shape $\mathfrak{a}$, a single $Z_{i,j,k}$ is equal to a Bernoulli trial over the complementary CDF of the Lomax distribution.
                                                                \begin{eqnarray}                         
                                                                E[Z_{i,j,k}] &=&  (1+\frac{i-j}{\mathfrak{b}})^{-(\mathfrak{a})}\\
                                                                \end{eqnarray}
                                                                It therefore holds 
                                                                \begin{eqnarray}                         
                                                                E[D_i]&\leq& E[D_{i-1}] (  1 -  \Pi_{j=0}^{i-1} (1-(1+\frac{i-j}{\mathfrak{b}})^{-(\mathfrak{a})})^{ \frac{a}{b} \frac{\alpha}{\beta}})\\
                                                                E[D_1] &=& (1-(1-(1+\frac{1}{\mathfrak{b}})^{-(\mathfrak{a})})^{C})
                                                                \end{eqnarray}
                                                                \item It now directly follows 
                                                                \begin{eqnarray}                         
                                                                V[Z_{n,i,1}] &=&  E[V[Z_{n,i,1}|\mu]] + V[E[Z_{n,i,1}|\mu]]\\
                                                                &=& E[E[Z_{n,i,1}|\mu]] - E[E[Z_{2n,2i,1}|\mu]] +       V[e^{(i-n)\mu}]\\
                                                                &=&  \frac{\mathfrak{b}^\mathfrak{a}}{(n+ \mathfrak{b}-i)^{\mathfrak{a}} } - \frac{\mathfrak{b}^\mathfrak{a}}{(2n+ \mathfrak{b}-2i)^{\mathfrak{a}} }\\
                                                                && +     \frac{\mathfrak{b}^\mathfrak{a}}{(2n+ \mathfrak{b}-2i)^{\mathfrak{a}} } - \frac{\mathfrak{b}^\mathfrak{a}}{(n+ \mathfrak{b}-i)^{\mathfrak{a}} } \frac{\mathfrak{b}^\mathfrak{a}}{(n+ \mathfrak{b}-i)^{\mathfrak{a}} }\\
                                                                &=&  \frac{\mathfrak{b}^\mathfrak{a}}{(n+ \mathfrak{b}-i)^{\mathfrak{a}} } (1 - \frac{\mathfrak{b}^\mathfrak{a}}{(n+ \mathfrak{b}-i)^{\mathfrak{a}} })
                                                                \end{eqnarray}
                                                                \begin{align}                    
                                                                &CE[  V[ Z_{n,i,k}|\mu]] +C^2 V[E[Z_{n,i,k}|\mu]] \\
                                                                =& C (E[E[Z_{n,i,1}|\mu]] -  E[E[Z_{2n,2i,1}|\mu]])+      C^2 V[e^{(i-n)\mu}]\\
                                                                =& C (\frac{\mathfrak{b}^\mathfrak{a}}{(n+ \mathfrak{b}-i)^{\mathfrak{a}} } -  \frac{\mathfrak{b}^\mathfrak{a}}{(2n+ \mathfrak{b}-2i)^{\mathfrak{a}} })+     C^2 (\frac{\mathfrak{b}^\mathfrak{a}}{(2n+ \mathfrak{b}-2i)^{\mathfrak{a}} } -  (\frac{\mathfrak{b}^\mathfrak{a}}{(n+ \mathfrak{b}-i)^{\mathfrak{a}} })^2)\\
                                                                =& C \frac{\mathfrak{b}^\mathfrak{a}}{(n+ \mathfrak{b}-i)^{\mathfrak{a}} } (1 - C \frac{\mathfrak{b}^\mathfrak{a}}{(n+ \mathfrak{b}-i)^{\mathfrak{a}} }) + (C^2-C) \frac{\mathfrak{b}^\mathfrak{a}}{(2n+ \mathfrak{b}-2i)^{\mathfrak{a}} }
                                                                \end{align}

                                                                \item For $M_i$ it holds by the same argument, 
                                                                \begin{eqnarray}
                                                                E[M_n] &=& \frac{\mathfrak{b}^\mathfrak{a}}{(\Delta n+ \mathfrak{b})^{\mathfrak{a}} }\\                      
                                                                V[M_n] &=& \frac{\mathfrak{b}^\mathfrak{a}}{(\Delta n+ \mathfrak{b})^{\mathfrak{a}} } (1-\frac{\mathfrak{b}^\mathfrak{a}}{(\Delta n+ \mathfrak{b})^{\mathfrak{a}} })
                                                                \end{eqnarray}

                                                                \end{itemize}

                                                                \end{proof}

        \section{Importance Sampling} \label{AP:SAMP}

        Importance sampling is a technique that, instead of drawing samples $r$ from the \emph{nominal sampling distribution} $p$ in order to estimate the expected value of some feature of the samples $f$, it draws the samples from an \emph{importance distribution} $q$. These samples are then weighted  according to the ratio between both distributions in order to obtain an estimate of $E[f]$ with a lower variance. This can vastly reduce the number of samples required to make statements with high confidence. 
        It is well known (see \citeauthor{kahn1953methods} \citeyear{kahn1953methods}) that the optimal importance density satisfies
        \begin{align}
        q(r) &=   \frac{f(r)p(r)}{ E[f(r)]}\label{eq:optq}
        \end{align}
        While calculating this exactly would require knowledge about the very value we want to estimate, it can often be approximated reasonably well. In our case, where each simulation run $r$ depends on tens of thousands of random variables, we define a heuristic measure that roughly indicates how likely a run is to fail and partition the set of all possible runs into buckets using this measure. We then approximate the optimal $q$ on the level of buckets.  

   As a first step, we now define the heuristic measure we use: 
%
%
%
        
                \begin{definition}
                                        For a deployment $x$, denote by 
                                        \begin{align}
                                        i^x= E[L^x_n]+\sqrt{(1-0.01)/0.01 *Var[L_n^x]}
                                        \end{align}
                        the upper bound of the $99\%$ confidence interval of a deployment's size in timestep $n$ as given by Cantelli's inequality.
                        For a given run $r$, pre-draw the parameters of all deployments that might arrive during the simulation run. 
                        Then consider the following extremely simplified simulation:
                        \begin{enumerate}
                                \item On the first of each month, 730 new deployments arrive. 
                                \item Deployments only die at the end of the month after their maximum lifetime is reached. They do not die when they reach zero cores. 
                                \item The cluster knows each deployment's exact type.
                                \item Deployment always take exactly their expected size.
                                \item Deployments are accepted into the cluster whenever $\sum_{\tilde{x}\in A_\pi (b)} i^x <22000.$  
                        \end{enumerate}
                        Denote by $X(n)$, $n \in [1,36]$ the set of deployments in the cluster at the beginning of each month during this simplified simulation. 
                        Then the \emph{badness measure} $BM$ of a run $r$ is defined as                   
                        \begin{align}
                        BM(r) = max_{n} \sum_{\tilde{x}\in X(n)} i^x. 
                        \end{align}                             
                \end{definition}   

        This is a reasonable (though highly heuristic) predictor of whether a run might produce a very large number failures. Most importantly, because it assumes away all randomness that occurs during the simulation run, it can be evaluated very quickly ($<1$ second). 
        
To properly utilize importance sampling, we now sort any simulation run $r$ into one of three buckets based on their $BM$ value: $I_1 = \left\lbrace r: BM(r)\leq 25000 \right\rbrace$,\\
 $I_2 = \left\lbrace r: 25000 \leq BM(r)\leq 30000\right\rbrace$, $I_3 = \left\lbrace r: 30000 \leq BM(r) \right\rbrace$. Before we can apply importance sampling, we calculate the probability for a given run to be in each of the buckets. For this, we calculate $BM$ for $100,000$ runs. The resulting probabilities can be found in Table \ref{tab:naivProb}.
                \begin{table}[t]

                        \centering
                
                        \vspace{2pt}
                        \begin{tabular}{rccc}
                                \noalign{\smallskip} \hline \noalign{\smallskip}
                                
                                \textbf{Probabilities} &      $I_1$ & $I_2$ &$I_3$\\
                                \hline \noalign{\smallskip}
                                  $p(I_i)$        & $0.5699 $  & $0.4121$&$0.018$ \\

                                \noalign{\smallskip} \hline \noalign{\smallskip}
                                \noalign{\smallskip} \hline \noalign{\smallskip}
                                
                                 $p(I_i|\cap_{k\geq i}I_k)$&   $0.88319 $  & $0.9582$&$1$ \\

                                \noalign{\smallskip} \hline \noalign{\smallskip}
                                \noalign{\smallskip} \hline \noalign{\smallskip}
                                $p_r(I_i|\cap_{k\geq i}I_k)$ &   $0.5369$  & $0.8816$&$0$  \\
                                \hline \noalign{\smallskip
                                }
                        \end{tabular}
                        \caption{Estimation of BM probabilities}
                                        \label{tab:naivProb}
                \end{table}

         To sample runs from the different buckets with different weights, we employ a type of rejection sampling: Before starting a simulation run $r$, we evaluate $BM(r)$. Depending on the bucket $I_i$ the run would result in, we then redraw with some probability $p_r(I_i)$ (i.e., all deployment parameters are discarded and redrawn) and block off all lower buckets (i.e., automatically rejecting any further redraws that would result in $I_j$, $j<i$). The highest bucket (in our case $I_3$) never gets redrawn, i.e., $p_r(I_3)=0$. This scheme continues iteratively until we accept a run. This results in the following importance distribution.
         
         \begin{proposition}
        For a run $r$ with nominal probability $p(r)$ and  BM such that $r\in I_i$, the above rejection scheme results in importance distribution $q$ with 
                           \begin{align}
                           q(r) &= p(r|I_i) \frac{p(I_i|\cap_{k\geq i}I_k) (1-p_r(I_i|\cap_{k\geq i}I_k))}{1-p(I_i|\cap_{k\geq i} I_k)p_r(I_i|\cap_{k\geq i}I_k)} \Pi_{j<i} \frac{p(\cap_{k>j} I_k  | \cap_{k\geq j} I_k )}{1-p(I_j | \cap_{k\geq j} I_k )p_r(I_j | \cap_{k\geq j} I_k)}\label{eq:q}          
                           \end{align}
         
         \end{proposition} 
         \begin{proof}
 
                We  first show this for two buckets $I_1$ and $I_2$. With only two buckets, since $I_2$ is the highest bucket, it has an acceptance probability of $1$.

                        It follows that for a run that would be in $I_1$, we redraw with probability $p_r(I_1)$ and otherwise accept. It thus holds
                  \begin{align}
                  q(I_1) &= p(I_1) ((1-p_r(I_1)) + p_r(I_1) q(I_1).
                  \end{align}
                  
                  This is a geometric series and therefore 
                  
                 \begin{align}
                  q(I_1) &=  \sum_{k=0}^\infty p(I_1) (1-p_r(I_1)) ( p(I_1)p_r(I_1))^k\\
                   &=  \frac{p(I_1) (1-p_r(I_1))}{1-p(I_1)p_r(I_1)}
                 \end{align}
                 Similarly, it holds 
                                           \begin{align}
                                           q(I_2) &=   \sum_{i=0}^\infty p(I_2) ( p(I_1)p_r(I_1))^i\\
                                           &=  \frac{p(I_2)}{1-p(I_1)p_r(I_1)}
                                           \end{align}
                                           
                                           Iteratively applying this argument to more than two buckets by dividing the top bucket into two buckets then yields 
                                                                   \begin{align}
                                                                   q(I_i) &= \frac{p(I_i|\cap_{k\geq i}I_k) (1-p_r(I_i|\cap_{k\geq i}I_k))}{1-p(I_i|\cap_{k\geq i} I_k)p_r(I_i|\cap_{k\geq i}I_k)} \Pi_{j<i} \frac{p(\cap_{k>j} I_k  | \cap_{k\geq j} I_k )}{1-p(I_j | \cap_{k\geq j} I_k )p_r(I_j | \cap_{k\geq j} I_k)}      
                                                                   \end{align}
                           Finally, by Bayes' theorem it holds that $q(r) = q(r|I_i) q(I_i)$ and since it further holds that $q(r|I_i) = p(r|I_i)$, the statement of the proposition follows.                                                                   
         \end{proof}
         
         To find good rejection probabilities $p_r$ that result in low sample variance, we did 500 runs for each bucket under the second moment policy with threshold 15000 to get a very rough estimate of $f$ (i.e. the scale out failure probability) for each bucket. The rejection probabilities $p_r$ are then calculated by combining Equation \eqref{eq:optq} and Equation \eqref{eq:q}. The resulting values can also be found in Table \ref{tab:naivProb}. 
         
\section{Ablation of Policies Without Marginal Heuristic}\label{AP:AB}
In this section, we ablate the simulation results for our policies with marginal heuristic with the same policies without these heuristic. Note again that without individual prior observations, no arriving deployment is ever marginal in the policy horizon. Thus, without prior observations, the marginal policies are equivalent to the policies without marginal heuristic. 
        \begin{figure}[t]
                \centering%
                \includegraphics[height=7cm]{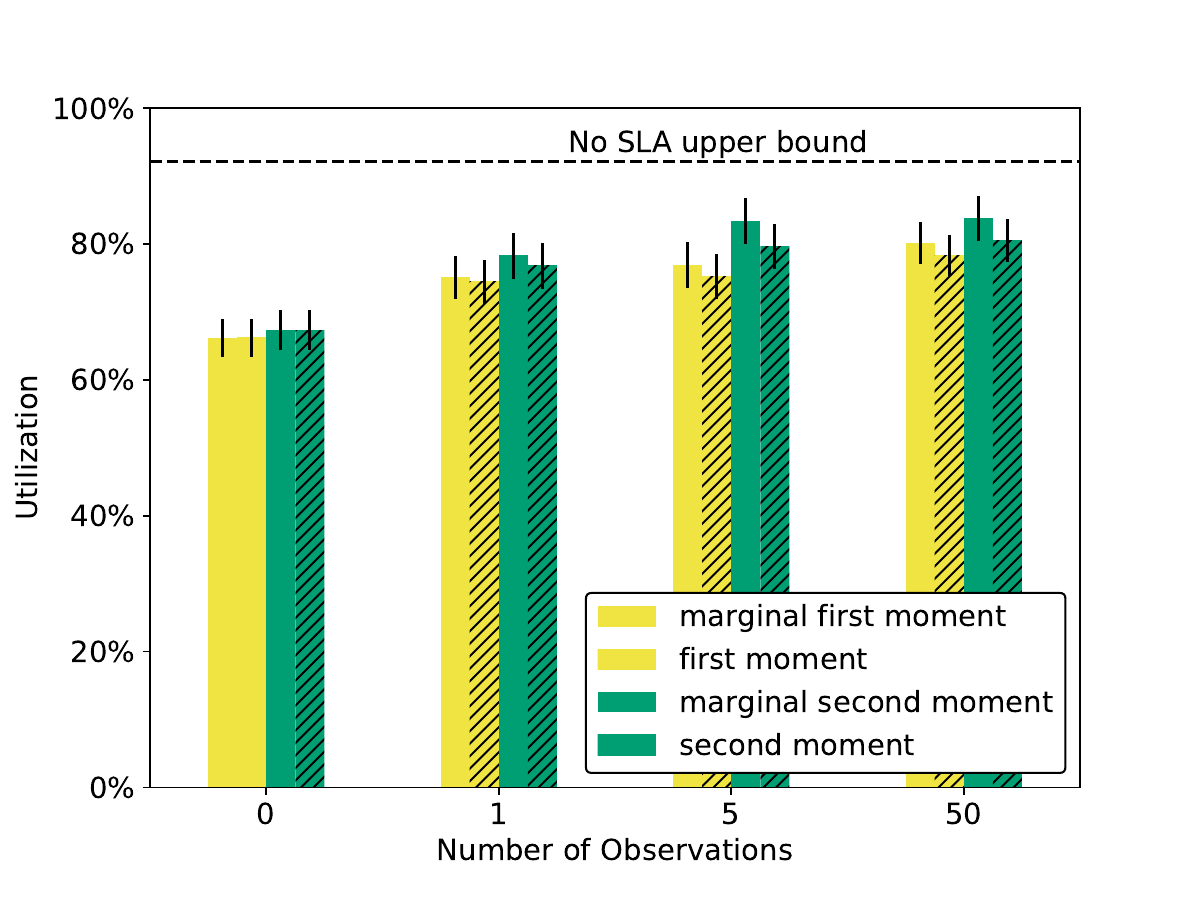}
                \caption{Ablation of policies with and without marginal heuristic (error bars indicate $95\%$ bootstrap confidence intervals)}
                \label{simref_ablation}
        \end{figure}%

We have simulated the policies without marginal heuristic for $1$, $5,$ and $50$ observations, with the same setup as in Section \ref{sim1}. In Figure \ref{simref_ablation}, we contrast those results with the results for the marginal policies.
We see that with just a few observations, the policies with and without marginal heuristic have a very similar performance (though the marginal heuristic still enables slightly higher utilization). This is not surprising, since relatively few arriving deployments are marginal at this level of prior information. Consequently, the utilization gap increases with more information. Both with $5$ and $50$ observations, the marginal second moment policy obtains a sizable utilization increase of more than $3\%$ compared to the second moment policy without the heuristic.

\vskip 0.2in

\end{document}